\documentclass[sigconf]{acmart}

\usepackage{booktabs} % For formal tables
\usepackage{float}
\usepackage{bm}
\usepackage{amsmath}
\usepackage{amsthm}
\usepackage{amssymb}
\usepackage{algorithm}
\usepackage{algorithmic}
\usepackage{color}
\definecolor{shadecolor}{rgb}{0.878906, 0.878906, 0.878906}

\setcopyright{rightsretained}
\newcommand{\textred}[1]{{#1}}
\acmDOI{10.475/123_4}

\acmISBN{123-4567-24-567/08/06}

\acmConference{ACM Sigmetrics conference}{June 2019}{Phoenix, AZ USA}
\acmYear{2019}
\copyrightyear{2019}

\acmArticle{1}
\acmPrice{15.00}

\begin{document}
%\title{Online Algorithm Design Framework for Inventory Constrained Revenue Maximization Problem}
% \title{Online Inventory Constrained Optimization with \\ Applications to Generalized One-way Trading}
\title{Competitive Online Optimization under Inventory Constraints}
% \title{Competitive Online Inventory Constrained Optimization}
%\titlenote{The first two authors contribute equally to the work.}
%\subtitle{Under submission}
%\subtitlenote{The full version of the author's guide is available as
%  \texttt{acmart.pdf} document}

 \author{Qiulin Lin, Hanling Yi}
 \authornote{The first two authors contribute equally to the work.}
 \affiliation{%
  \department{Information Engineering}
  \institution{The Chinese University of Hong Kong}
 }

%\author{Qiulin Lin}
%\authornote{The first two authors contribute equally to the work.}
%\affiliation{%
% \department{Information Engineering}
% \institution{The Chinese University of Hong Kong}
%}
%
%\author{Hanling Yi}
%\affiliation{%
% \department{Information Engineering}
% \institution{The Chinese University of Hong Kong}
%}

\author{John Pang}
\affiliation{%
 \department{Computing and Math. Sciences}
 \institution{California Institute of Technology}
}

\author{Minghua Chen}
\authornote{Corresponding author.}
\affiliation{%
 \department{Information Engineering}
 \institution{The Chinese University of Hong Kong}
}

\author{Adam Wierman}
\affiliation{%
 \department{Computing and Math. Sciences}
 \institution{California Institute of Technology}
}

\author{Michael Honig}
\affiliation{%
 \department{Electrical Engineering}
 \institution{Northwestern University}
}

\author{Yuanzhang Xiao}
\affiliation{%
 \department{Hawaii Center for Adv. Comm.}
 \institution{University of Hawaii at Manoa}
}

\renewcommand{\shortauthors}{Q. Lin, H. Yi, J. Pang, M. Chen, A. Wierman, M. Honig, and Y. Xiao}

%%%%%%%%%%%%%%%%%%%%%%%%%%%%%% Textclass specific LaTeX commands.
\theoremstyle{plain}
\newtheorem{thm}{\protect\theoremname}
\theoremstyle{plain}
\newtheorem{defn}[thm]{\protect\definitionname}
\theoremstyle{plain}
\newtheorem{prop}[thm]{\protect\propname}
\theoremstyle{plain}
\newtheorem{lem}[thm]{\protect\lemmaname}
\newcommand\txtred[1]{{\color{red}#1}}
\newcommand\hanling[1]{{\color{blue}#1}}
\providecommand{\definitionname}{Definition}
\providecommand{\propname}{Proposition}
\providecommand{\lemmaname}{Lemma}
\providecommand{\theoremname}{Theorem}

\begin{abstract}
This paper studies online optimization under inventory (budget) constraints. While online optimization is a well-studied topic, versions with inventory constraints have proven difficult. We consider a formulation of inventory-constrained optimization that is a generalization of the classic one-way trading problem and has a wide range of applications.  We present a new algorithmic framework, \textsf{CR-Pursuit}, and prove that it achieves the minimal competitive ratio among all deterministic algorithms (up to a problem-dependent constant factor) for inventory-constrained online optimization. 
Our algorithm and its analysis not only simplify and unify the state-of-the-art results for the standard one-way trading problem, but they also establish novel bounds for generalizations including concave revenue functions. For example, for one-way trading with price elasticity %where few was previously reported
, the \textsf{CR-Pursuit} algorithm achieves a competitive ratio that is within a small additive constant (i.e., 1/3) to the lower bound of $\ln \theta+1$, where $\theta$ is the ratio between the maximum and minimum base prices.

\end{abstract}




\maketitle

\section{Introduction}

Online optimization is a foundational topic in a variety of communities, from machine learning to control theory to operations research. There is a large and active community studying online optimization in a wide range of settings, both looking at theoretical analysis and real-world applications. The applications of online optimization are wide ranging, e.g., multi-armed bandits \cite{schlag1998imitate,auer1995gambling,bubeck2012regret}, network optimization (with packing constraints) \cite{guo2016dynamic,guo2018joint}, data center capacity management \cite{lin2013dynamic, lu2013simple, ren2018datum}, smart grid control \cite{lu2013online, zhang2018peak, pang2017optimal}, and beyond.  Further, a diverse set of algorithmic frameworks have been developed for online optimization, from the use of classical potential functions, e.g., \cite{grove1991harmonic,abernethy2009competing}, to primal-dual techniques, e.g., \cite{shalev2009mind,guo2016dynamic}, to approaches based on receding horizon control, e.g., \cite{michalska1993robust,mattingley2011receding}. Additionally, many variations of online optimization have been studied, e.g., online optimization with switching costs \cite{lin2012online1,bansal20152,li2018using}, online optimization with predictions \cite{lu2013online,chen2015online,li2018using}, convex body chasing \cite{friedman1993convex,antoniadis2016chasing,bansa2018nested}, and more.%\textbf{ADD MORE CITATIONS TO THIS PARAGRAPH.  HAVE AT LEAST TWO FOR EACH TOPIC...ONE OLD AND ONE NEW}

In this paper, we focus on an important class of online optimization problems that has proven challenging: \emph{online optimization under inventory (budget) constraints (OOIC)}.  In these problems a decision maker has a fixed amount of inventory, e.g., airlines selling flight tickets or battery owners participating in power contingency reserves market, and must make a decision in each of the $T$ rounds with the goal of optimizing per-round revenue functions.  The challenge is that the decision maker does not have knowledge of future revenue functions or when the final round will occur, i.e., the value of $T$. Further, the strict inventory constraint means that an action now has consequences for future rounds.  As a result of this entanglement, positive results have only been possible for inventory constrained online optimization in special cases to this point, e.g., the one-way trading problem~\cite{el2001optimal}. 

More formally, a decision maker in an OOIC participates in $T$ rounds, without knowing $T$ ahead of time.  In each round, the decision maker selects an action $v_t\geq 0$, e.g., an amount to sell,  after observing a concave revenue function $g_t(\cdot)$.  Though the decision maker observes the revenue function each round before choosing an action, it is typically not desirable to choose an action to maximize the revenue in each round due to the limited inventory $\Delta$.  Specifically, the actions are constrained by $\sum_{t=1}^T v_t\leq \Delta$, and consequently an action taken at time $t$ constrains future actions.  In particular, if the inventory is used too early then better revenue functions may appear later, when inventory is no longer available.

%Unlike in the case for airlines, in our inventory constrained setting, we do not allow for the violation of this inventory constraint. 
%We also do not allow for two-way trading, but \emph{adopt the one-way trading setting whereby the inventory, when consumed or sold, cannot be replenished}. In the power contingency reserves market example, this correspond to the assumption that batteries are charged in the night and do not get charged across the day. Lastly, the \emph{online} part of the problem manifests as an uncertainty in the incoming revenue functions and in the unknown stopping time. This uncertainty assumption is common in online algorithm literature~\cite{el2001optimal,yang2017online}.
%This uncertainty can otherwise manifest in other ways, e.g., the packing constraints \cite{buchbinder2009design}, or the stopping time as in the one-way trading problem with unknown duration \cite{el2001optimal}. 
%These three components of our problem, the (i) fixed inventory, the (ii) one-way trading, and the (iii) online uncertainties in stopping time and revenue functions, together make the online optimization under inventory constraint problem challenging.

OOIC generalizes many well-known online learning and revenue maximization problems.  One of the most prominent is the one-way trading problem \cite{el2001optimal}, where a trader owns some assets (e.g., dollars) and aims to exchange them into other assets (e.g., yen) as much as possible, depending on the price (e.g., exchange rate). There is a long history of work on one-way trading \cite{el2001optimal,chin2015competitive,damaschke2009online,fujiwara2011average,lorenz2009optimal,zhang2012optimal}, as we describe in Sec. \ref{application}, and OOIC includes both the classic one-way trading problem and variations with concave revenue functions and price elasticity.  

\textred{\textbf{Applications. }}Beyond the one-way trading problem, OOIC also captures a variety of other applications.  Three examples that have motivated our interest in OOIC are (i) power contingency reserve markets \cite{Shafiee2018Battery,Akhavan2014Storagy}, (ii) network spectrum trading \cite{Bogucka2012Spectrum,Qian2011Spectrum}, and \textred{(iii) online advertisement}.

In power contingency reserve markets, the system operator faces a contingency, e.g., shortfall of supply that may lead to cascading blackouts, and communicates this need to either supplement the power system using battery or cut down large scale power supply. Consider the perspective of a battery supply owner that is deciding when to take part in a contingency. A contingency may be solved immediately, or it may instead cause a larger contingency whereby the system operator is willing to pay more at a later time epoch. In preparation to participate in these contingencies, batteries are charged earlier and therefore the marginal cost of participation only manifests as an opportunity cost against future participation in the day. These situations highlight the need for the online properties considered in our work: (i) the unknown ending time $T$, (ii) future revenue functions are not known, and (iii) a costless, strict inventory constraint. 

Similarly, in spectrum trading, the owner of a spectrum band sells bandwidth to make sure that profit or revenue is maximized given the investments that have already been made to procure the particular bandwidth.  This means that any cost with regards to sales only appears as opportunity cost against future possible sales. Similarly, a potential buyer who is turned down may seek bandwidth from a different provider, and may never return, or situations may change between time epochs, highlighting the same three properties as before: (i) the unknown ending time $T$, (ii) future revenue functions are not known, and (iii) a costless and strict inventory constraint. 

\textred{In online advertisement, an advertiser with a given budget would like to invest into keywords from Internet search engines, e.g., Google AdWords. Potential keywords come in an online fashion and may be unavailable at any time. It has also been shown in \cite{Devanur2012onlinead} that revenue can be modelled as a concave function with respect to the investment. The advertiser needs to decide how to invest its budget for keywords to maximize the overall revenue, once again highlighting the same three properties listed above.}

\textbf{Contributions.}  In this paper we develop a new algorithmic framework, called \textsf{CR-Pursuit}, and apply it to develop online algorithms for the OOIC problem with the optimal competitive ratio (up to a problem-dependent constant factor).  Further, we prove that \textsf{CR-Pursuit} %achieves a nearly optimal competitive ratio and 
provides the first positive results for a generalization of the classic one-way trading problem with concave revenue functions and price elasticity.  In more detail, we summarize our contributions as follows.

First, we introduce a new algorithmic framework, \textsf{CR-Pursuit}, in Sec.~\ref{sec:online_alg_framework}. The framework is based on the idea of ``pursuing'' an optimized competitive ratio at all time. The framework is parameterized by a tight upper bound on the competitive ratio, which is then ``pursued'' with the actions in each round. We apply the framework to OOIC and generalizations of the one-way trading problem in this paper, but the framework has the potential for broad applicability beyond these settings as well. Along the way, we also derive several useful results on the offline optimal solution in Sec.~\ref{sec:offline}, which may be of independent interest.
%by solving a characteristic function, which relies on mathmetically characterization of the worst case inventory needed to maintain the target competitive ratio over all possible input sequences.

Second, in Sec.~\ref{UB_CR}, we apply \textsf{CR-Pursuit} to the OOIC problem to achieve the optimal competitive ratio among all deterministic algorithms (up to a problem-dependent constant factor). To obtain these bounds we use two technical ideas that are of general interest beyond OOIC.  First, we prove that it suffices to focus on the single-parametric \textsf{CR-Pursuit} algorithm for achieving the optimal competitive ratio, thus significantly reducing the search space of optimal online algorithms. Second, we identify a ``critical'' input sequence that highlights an important structural property of the space of input sequences. By applying \textsf{CR-Pursuit} to this critical sequence, we characterize a lower bound on the optimal competitive ratio as $\ln \theta +1$ where $\theta$ is the ratio between the maximum and minimum base prices to be defined in Sec.~\ref{ProblemFormulation}. Subsequently, for any other input, the performance ratio achieved by \textsf{CR-Pursuit} is upper bounded by the product of a problem-dependent factor and the lower bound. This structure not only suggests a principled approach to characterizing the optimal competitive ratio, but also immediately shows that \textsf{CR-Pursuit} achieves the optimal competitive ratio (up to a problem-dependent factor) among all deterministic algorithms.

%Specifically, we prove a lower bound on the competitive ratio of any deterministic algorithm of $\ln\theta+1$, where $\theta$ is the ratio between the maximum and minimum base prices across market at different time epochs. This lower bound is within a factor of $c$ of the upper bound we prove on \textsf{CR-Pursuit}, where $c$ is a constant that depends on the gradient properties and the maximizers of the revenue functions. 

Third, we apply \textsf{CR-Pursuit} to one-way trading problems in Sec.~\ref{application}. The novel framework simplifies and unifies the state-of-the-art results of the classic one-way trading problem. In particular, the critical input discussed above is simply the worst case one for classical one-way trading; hence, \textsf{CR-Pursuit} achieves the optimal competitive ratio $\ln \theta+1$. Further, we show that \textsf{CR-Pursuit} performs well for generalizations of one-way trading where no positive results were previously known.  Specifically, for one-way trading with price elasticity and concave revenue functions, \textsf{CR-Pursuit} achieves a competitive ratio that is 
% in $[\ln \theta+5/4,\ln\theta + 4/3]$, which is 
within a small additive constant (i.e., $1/3$) to the general lower bound of $\ln \theta+1$.

\section{Related work.}\label{related_work}
Online optimization is a large and rich research area and excellent surveys can be found in \cite{albers2003online,fiat1998online}.  Well-known problems in the online optimization paradigm include the classic secretary problem \cite{chow1964optimal}, the ski rental problem \cite{karlin1988competitive}, the one-way trading problem \cite{el2001optimal}, and the $k$-server problem \cite{fiat1990competitive}. Our results represent the most general results to date for a situation where actions are subject to a fixed inventory constraint. 

% The proposed CR-Pursuit framework not only simplifies and unifies the state-of-the-art results for the standard one-way trading problem, but also establishes novel bounds for generalizations including concave revenue functions. 

The problem considered here is a generalization of the classical one-way trading problem, which has received considerable attention, e.g., \cite{el2001optimal,chin2015competitive,damaschke2009online,fujiwara2011average,lorenz2009optimal,zhang2012optimal}. In the one-way trading problem an online decision maker is sequentially presented with exchange rates within a bounded region, and she desires to trade all her assets to another. The amount of assets traded in a single time period is assumed to be small enough to not affect the eventual price. El-Yaniv et. al.~\cite{el2001optimal} propose a threshold-based online algorithm with competitive ratio $O(\ln\theta)$. Any remaining items must be sold at the last epoch as that is revenue maximizing. On the other hand, our analysis allows for leftover inventory (since selling all assets at the last time step may not be revenue maximizing solution for the last time step in the presence of price elasticity or concave revenue functions) and an unknown stopping time, while retaining the competitive ratio. 

Variants of the one-way trading problem have
been studied in the literature. Chin et al. \cite{chin2015competitive} and Damaschke et al.~\cite{damaschke2009online} study the one-way trading problem with unbounded prices and time-varying price bounds, respectively. Zhang et al. \cite{zhang2012optimal} study the problem when every two consecutive prices are interrelated. Fujiwara et al. \cite{fujiwara2011average} study the problem using average-case competitive
analysis under the assumption that the distribution of the maximum
exchange rate is known. Kakade et al. \cite{kakade2004competitive}
incorporate market volume information and study another one-way trading
model in the stock market, called the price-volume trading problem. 
While the classical one-way trading problem mostly deals with linear revenue functions, we note that in our problem we consider general concave revenue functions, which allows us to capture a boarder class of interesting settings, e.g.,  one-way trading with price elasticity.

Beyond the one-way trading problem, OOIC is also highly related to generalizations of the secretary problem and prophet inequalities, e.g., \cite{rubinstein2016beyond,feldman2018submodular, babaioff2008online}. %\textbf{ADD CITATIONS}.  
Strong positive results have been obtained for these problems; however the analytic setting considered differs dramatically from the current paper.  Specifically, we consider a worst case analysis whereas analysis of the secretary problem and prophet inequalities focus on stochastic instances.  Under the stochastic setting, so-called ``thresholding'' algorithms are effective; however such algorithms have unbounded competitive ratios in the worst case setting, even under the simplest assumptions.  

Prior to this work, the most general results known for online problems with inventory constraints are for the class of problems termed online optimization with packing constraints, e.g., \cite{buchbinder2005online,buchbinder2009design,Azar2016convex,arora2012multiplicative,bansal2012primal}.  This stream of work developed an interesting algorithmic framework based on a primal-dual or multiplicative weights update approaches, which centers around maintaining a dual variable for each constraint, understood as a shadow (or pseudo) price for the constraint given the information thus far.  While the inventory constraints we consider are packing constraints, our formulation is fundamentally different than the formulation considered in these papers.  In these papers, the constraints come in an online fashion; whereas in our work, the revenue functions arrive in an online fashion. 
% and the stopping time is unknown beforehand. 

Another related online optimization problem is the $k$-search problem, where a player searches for the $k$ highest prices in a sequence that is revealed to her sequentially. When $k\to \infty$, the $k$-max search problem becomes the one-way trading problem~\cite{lorenz2009optimal}. Lorenz et. al.~\cite{lorenz2009optimal} propose optimal deterministic and randomized online algorithms for both the $k$-max search and $k$-min search problem. This is different from the well-known $k$-server problem, where an online algorithm must control the movement of $k$ servers in a metric space to minimize the movement (or latency involved) in serving future requests. A popular algorithmic framework for the $k$-server problem is the \emph{potential function framework}. In contrast to our CR-Pursuit approach, the potential function approach requires a bound between the offline optimal cost and the online cost at each time epoch with respect to the potential. 

Finally, it is important to distinguish our work from the literature studying \emph{regret} in online  optimization, e.g.,   \cite{hazan2007logarithmic,chen2015online}.  While regret is a natural measure for many online optimization problems, when inventory constraints are present it is no longer appropriate to compare against the best static action, as is done by regret.  Static actions are poor choices when optimizing revenue subject to inventory constraints.  Instead, competitive ratio is the most appropriate measure.  Further, note that there is a fundamental algorithmic trade-off between optimizing regret and competitive ratio, even when inventory constraints are not present.  In particular, \cite{andrew2013tale} shows that no algorithm can obtain both sub-linear regret and constant competitive ratio.  

\section{Problem Formulation}
\label{ProblemFormulation}
We study an online optimization problem where a decision maker sells inventory across an interval of discrete time slots in order to maximize the aggregate revenue. The revenue functions of individual slots are revealed sequentially in an online fashion, and the interval length is unknown to the decision maker. The initial inventory is given in advance as a constraint, and it bounds the decision maker's aggregate selling quantities across time slots. \footnote{We emphasize that in contrast to the works on online optimization with packing constraints \cite{buchbinder2005online,buchbinder2009design,Azar2016convex,arora2012multiplicative,bansal2012primal}, the uncertainty in our optimization problem is not that the inventory constraint is unknown beforehand, but rather the revenue functions arrives in an online fashion. } The key notations used in this paper are summarized in Tab.~\ref{tab:TableOfNotations}. Throughout this paper, we use $[n]$  to represent the set $\{1,2,...,n\}$ where $n$ is a positive integer.
%The setting that knowing the revenue function at the beginning of a slot
%is a common assumption in online algorithm literature, see for example
%\cite{el2001optimal,zhang2012optimal,chin2015competitive}.

%\begin{table}[h!]\caption{Summary of Notations.}
%\centering % to have the caption near the table
%\begin{tabular}{c|p{10cm} }
%\toprule
%$T$  & The number of time slots\\
%\hline
%$\Delta$  &  The initial inventory \\
%\hline
%$g_t(v)$  & The revenue function of  time slot $t$ \\
%\hline
%$\sigma^{[1:t]}$ & Input (revenue function) sequence up to time $t$, i.e., $\{g_1,g_2,..., g_t\}$ \\
%\hline
%$p(t)$ & Base price at time $t$, i.e., $g'_t(0)$\\
%\hline
%$m,\,M$  & The lower and upper bounds of $p(t)$, $\forall t\in [T]$\\
%\hline
%$\theta$ & The ratio of $M/m$\\
%\hline
%$\lambda$ & The dual variable associated with the inventory constraint in OOIC \\
%\hline
%$v_t$  & The selling quantity at time $t$ \\
%\hline
%$\bar{v}_t$  & The selling quantity of CR-Pursuit($\pi$) at time $t$  \\
%\hline
%$v_t^*$  & The optimal selling quantity at time $t$ under the offline setting\\
%\hline
%$\hat{v}_t$ & A maximizer of $g_t(v)$ over $[0,\Delta]$  \\
%\hline
%$\Phi_\Delta(\pi)$ & The worst case (maximal) inventory over all possible sequences of inputs needed to maintain a competitive ratio $\pi\geq 1$ for \textsf{CR-Pursuit}($\pi$)\\
%\bottomrule
%\end{tabular}
%\label{tab:TableOfNotations}
%\end{table}

 \begin{table}[h!]\caption{Summary of Notations.}
 \centering % to have the caption near the table
 \begin{tabular}{c|p{6.5cm}}
 \toprule
 $T$  & The number of time slots\\
 \hline
 $\Delta$  &  The initial inventory \\
 \hline
 $g_t(v)$  & The revenue function of  time slot $t$ \\
 \hline
 $\sigma^{[1:t]}$ & Input (revenue function) sequence up to time $t$, i.e., $\{g_1,g_2,..., g_t\}$ \\
 \hline
 $p(t)$ & Base price at time $t$, i.e., $g'_t(0)$\\
 \hline
 $m,\,M$  & The lower and upper bounds of $p(t)$, $\forall t\in [T]$\\
 \hline
 $\theta$ & The ratio of $M/m$\\
 \hline
 $\lambda$ & The dual variable associated with the inventory constraint in OOIC \\
 \hline
 $v_t$  & The selling quantity at time $t$ \\
 \hline
 $\bar{v}_t$  & The selling quantity of CR-Pursuit($\pi$) at time $t$  \\
 \hline
 $v_t^*$  & The optimal selling quantity at time $t$ under the offline setting\\
 \hline
 $\hat{v}_t$ & A maximizer of $g_t(v)$ over $[0,\Delta]$  \\
 \hline
 $\Phi_\Delta(\pi)$ & The worst case (maximal) inventory over all possible sequences of inputs needed to maintain a competitive ratio $\pi\geq 1$ for \textsf{CR-Pursuit}($\pi$)\\
 \bottomrule
 \end{tabular}
 \label{tab:TableOfNotations}
 \end{table}

More specifically, at time $t\in[T]$, upon observing the revenue function  $g_t(\cdot)$, the decision maker has to make an irrevocable decision on an action (quantity) $v_t$. Upon choosing $v_t$ the decision maker receives a revenue of $g_t(v_t)$. The overall objective is to maximize the aggregate revenue, while respecting the inventory constraint $\sum_{t\in [T]}v_t \leq\Delta$.\footnote{The assumption that the action is chosen \emph{after} observing the function differs from the classical online convex optimization literature \cite{hazan2007logarithmic,li2012online}, but matches the literature on online convex optimization with switching costs \cite{lin2012online1,bansal20152,li2018using} and the literature on competitive algorithm design, including those on buy-or-rent decision making problems~\cite{karlin1988competitive,lu2013simple,zhang2018peak} and metrical task systems \cite{borodin1992optimal,fiat2003better,lu2013online}.  It allows an isolation of the inefficiency resulting from  inventory constraints rather than also including the inefficiency resulting from the of lack of knowledge of the function.} We assume that $g_t(\cdot),\forall t\in [T]$, satisfy the following conditions:
\begin{itemize}
    \item $g_t(v)$ is concave, increasing, and differentiable over $[0,\Delta]$;
    \item $g_t(0)=0$;
    \item $p(t)\triangleq g'_t(0)>0$ and $p(t)\in[m,M]$.
\end{itemize}

The first condition is a smoothness condition on the revenue function and a natural diminishing return assumption. It also limits our discussion in the more interesting setting where at each time, selling more could never decrease revenue. The second condition implies that selling nothing yields no revenue. The third condition limits the marginal revenue at the origin (named base price hereafter) and ensures that it is beneficial to sell, since the base price is positive. Denote the family of all possible revenue functions at time $t$ as $\mathcal{G}$. We assume $m$ and $M$ are known beforehand to the decision maker and denote $\theta=M/m$.

We formulate the problem of online optimization under inventory constraints (OOIC) as follows:
\begin{align}
\text{OOIC}:\quad\max\quad & \sum_{t=1}^{T}g_t(v_t)\\
\mbox{s.t.}\quad & \sum_{t=1}^{T}v_t\le\Delta,\label{eq:inventory_const}\\
\mbox{var.}\quad & v_t\ge0,\forall t\in[T]. \label{eq:v_t_geq_0_constraints}
\end{align}
\textred{Without loss of generality, we assume that the inventory constraint in \eqref{eq:inventory_const} is active at the optimal solution.}

We can interpret the inventory constraint (\ref{eq:inventory_const}) in an OOIC in a parallel way to the inventory constraint in the one-way trading problem \cite{el2001optimal}.  In particular, in the one-way trading problem the trader has to decide in each slot the selling quantity $v_t$ to maximize the total revenue at the stopping time $T$.  In fact, when setting the family of functions $\mathcal{G}$ to be the family of revenue functions of the form $g_t(v_t) = p(t)v_t$, we can see OOIC covers the  one-way trading problem as a special case. Additionally, when addressing revenue functions of the form $g_t(v_t) = v(t) (p(t) - f_t(v_t) )$ where $f_t$ is a convex function representing price elasticity, OOIC represents a generalized one-way trading problem with price elasticity.

To study the performance of an online algorithm for OOIC we use the \emph{competitive ratio} as the metric of interest.\footnote{Note that many papers in the online optimization literature, e.g., \cite{hazan2007logarithmic}, focus on \emph{regret} instead of competitive ratio, but regret is not an appropriate measure when inventory constraints are considered since static actions are no longer appropriate.  Our focus on competitive ratio matches that of the literature on secretary problems \cite{babaioff2008online, rubinstein2016beyond}, prophet inequalities \cite{hajiaghayi2007automated,rubinstein2016beyond}, online optimization with switching costs \cite{lin2012online1,lu2013online,lu2013simple,bansal20152,li2018using}, etc.}   Let $\mathcal{A}$ be a deterministic
online algorithm. It is called
$\pi$-\emph{competitive} if
\[
\pi=\max_{\sigma\in\Sigma}\;\frac{\eta_{OPT}(\sigma)}{\eta_{\mathcal{A}}(\sigma)},
\]
where $\Sigma$ is the set of all possible inputs ($g_t(\cdot)$, $t\in[T]$)
and $\eta_{OPT}(\sigma)$ and $\eta_{\mathcal{A}}(\sigma)$ are the
revenues generated by the optimal offline algorithm $OPT$ and the online algorithm $\mathcal{A}$, respectively. This value $\pi$ is the competitive ratio (CR) of the algorithm $\mathcal{A}$.

\section{Insights on the Offline Solution} \label{sec:offline}
In this section, we 
% characterize the  optimal offline solution of OOIC and 
derive several results on the optimal offline solution. They are useful in the design and analysis of our algorithmic framework \textsf{CR-Pursuit} in Sec.~\ref{sec:online_alg_framework}.

% The comparison against an offline optimal solution necessitates the analysis and understanding of the offline optimal solution of the OOIC problem, where the input sequence of functions $\sigma$ and stopping time $T$ are known beforehand to the decision maker.

% \subsection{An Offline Optimal Solution}

Under the offline setting where $g_t(\cdot)$, $\forall t\in[T]$, are known in advance to the decision maker, OOIC is a convex problem and can be solved efficiently. Let $v^*$ be the optimal primal solution and $\lambda^*$ be the optimal dual variable associated with the inventory constraint in \eqref{eq:inventory_const}. We note that $\lambda^*$ can be obtained by the algorithm in Alg.~\ref{alg:Offline-algorithm} in Appendix~\ref{apx:offline.optimal.alg}, based on a binary search idea. The following proposition gives a set of optimality conditions for the optimal primal solutions $v^*$ and the optimal dual variable $\lambda^*$.

%We propose a ``water-filling'' like algorithm to obtain the optimal revenue.
%, and study how revenue changes as the initial inventory increases.
%In the offline setting, problem $OOIC$ is a convex problem. 

% We associate Lagrangian multipliers $\lambda\ge0$ with the constraint in \eqref{eq:inventory_const} and $\mu_t\ge0$, $\forall t\in[T]$ with the constraints in \eqref{eq:v_t_geq_0_constraints}. The Lagrangian of the OOIC problem is then given by
% \[
% L\left(\bm{v},\lambda,\bm{\mu}\right)\triangleq\sum_{t=1}^{T}g_{t}(v_t)+\lambda\left(\Delta-\sum_{t=1}^{T}v_t\right)+\sum_{t=1}^{T}v_t\mu_t.
% \]

% By investigating the KKT conditions of the problem, we obtain an optimal offline algorithm summarized in Alg.
% \ref{alg:Offline-algorithm} in Appendix~\ref{apx:offline.optimal.alg}. Let $\left(\lambda^{*}, \mu^*\right)$ be the optimal dual solution obtained by Alg.~\ref{alg:Offline-algorithm}. 

\begin{prop}
\label{thm:offline_solution} Under our setting that the inventory constraint is active at the optimal solution, the optimal primal and dual solutions $v^*$ and $\lambda^{*}$ satisfy  (i) $\lambda^*\geq 0$ and $\sum_{t=1}^Tv^*_t=\Delta$ and (ii) for each $t\in[T]$,
\begin{align}
    & \begin{cases}
v_{t}^{*} = 0, & \mbox{if }g'_{t}\left(0\right)<\lambda^{*};\\
v_t^* \in V_t(\lambda^*)\triangleq\{v_t|g'_t(v_t)=\lambda^*,v_t\in[0,\Delta]\}, & \mbox{otherwise}.
\end{cases}
\end{align}
% either $v^{*}_t = 0$ or $g'_t(v^*_t)=\lambda^*$
%     \begin{eqnarray}
%         v^{*}_t = 0 \mbox{ or } \in V_t(\lambda^*) \mbox{ and }  \sum_{t=1}^Tv^*_t=\Delta. \label{eq:vt-1},
%     \end{eqnarray}
% where $V_t(\lambda^*)\triangleq \{v|g'_t(v)=\lambda^*,v\in[0,\Delta]\}$. Note $v^*_t=0$ if $V_t(\lambda^*)=\emptyset$.
%     \begin{eqnarray}
%         v^{*}_t \in V_t(\lambda^*) \mbox{ and }  \sum_{t=1}^Tv^*_t=\Delta. \label{eq:vt-1},
%     \end{eqnarray}
% where $V_t(\lambda^*)\triangleq \{v|g'_t(v)=\lambda^*,v\in[0,\Delta]\}$. Note $v^*_t=0$ if $V_t(\lambda^*)=\emptyset$.
% \begin{itemize}
%     \item if $\lambda^*=0$, then $v^{*}_t=\hat{v}_t$, where $\hat{v}_t$ is the optimizer of $g_t(\cdot)$ over $[0,\infty)$;
%     \item otherwise, define $V_t(\lambda^*)\triangleq \{v|g'_t(v)=\lambda^*,v\in[0,\Delta]\}$ as the set of $v$'s whose marginal revenue matches the marginal cost $\lambda^*$ at time $t$, then
%     \begin{eqnarray}
%         v^{*}_t \in V_t(\lambda^*),\mbox{ and }  \sum_{t=1}^Tv^*_t=\Delta. \label{eq:vt-1},
%     \end{eqnarray}
%     and $v^*_t=0$ if $V_t(\lambda^*)=\emptyset$.
% \end{itemize}
\end{prop}
Recall that at time $t$, the marginal revenue evaluated at $v_t$ is $g'_t(v_t)$, which is no larger than the base price $p(t)=g'_t(0)$ due to the concavity of $g_t(\cdot)$. The optimal dual variable $\lambda^{*}$ can be interpreted as the marginal cost (shadow price) of the inventory. Then Proposition~\ref{thm:offline_solution} says that, at the optimal solution, the marginal revenue must equal the marginal cost in the slots with positive selling quantities.   Moreover, it is optimal to sell only in the slots in which the base price is higher than the optimal marginal cost, i.e., $p(t)>\lambda^*$. These observations are similar to those in the Cournot competition literature, e.g., \cite{pang2017efficiency}.

Next, we reveal two interesting observations on the offline optimal aggregate revenue. Recall the input (revenue function) sequence until time $t$ is 
\[
    \sigma^{[1:t]}=\sigma^{[1:t-1]}\cup \left\{g_t(\cdot )\right\} = \sigma^{[1:t-2]}\cup \left\{g_{t-1}(\cdot )\right\} \cup \left\{g_t(\cdot )\right\} = \cdots .
\]
Recall that $\eta_{OPT}\left(\sigma^{[1:t]}\right)$ is the offline optimal aggregate revenue given the input $\sigma^{[1:t]}$. The following lemma bounds the increment of the optimal aggregate revenue as $t$ increases.

% 
% Recall $\sigma^{[1:t]}$ is the input (revenue function) sequence until time $t$. We apply the offline optimal algorithm to solve OOIC given $\sigma^{[1:t]}$. Let $\boldsymbol{\tilde{v}}$ and $\left(\lambda_t, \right)$ be the optimal primal and dual solutions that we obtain, respectively.

\begin{lem}
\label{lem:opt_difference_bound} 
Let  $\lambda_{t-1}$ and $\lambda_t$ be the optimal dual variables associated with the inventory constraint given the inputs $\sigma^{[1:t-1]}$ and $\sigma^{[1:t]}=\sigma^{[1:t-1]}\cup \left\{g_t(\cdot )\right\}$, respectively. Let $\tilde{v}_t$ be the optimal offline solution in the (last) time slot $t$ given the input $\sigma^{[1:t]}$. The following inequalities hold:
\begin{align}
\eta_{OPT}\left(\sigma^{[1:t]}\right)-\eta_{OPT}\left(\sigma^{[1:t-1]}\right) \geq & g_t\left(\tilde{v}_t\right)-\lambda_{t}\tilde{v}_t, \label{eq:off_Optimal_value_increment_lower_bound}
\end{align}
and 
\begin{align}
\eta_{OPT}\left(\sigma^{[1:t]}\right)-\eta_{OPT}\left(\sigma^{[1:t-1]}\right)\leq & g_t\left(\tilde{v}_t\right)-\lambda_{t-1}\tilde{v}_t \leq g_t\left(\hat{v}_t\right), \label{eq:off_Optimal_value_increment_upper_bound}
\end{align}
where $\hat{v}_t$ is the maximizer of $g_t(\cdot)$ over $[0,\Delta]$.
\end{lem}
Note that $\lambda_{t-1}$ and $\lambda_t$ are the marginal costs of inventory at the optimal solutions to OOIC given the inputs $\sigma^{[1:t-1]}$ and $\sigma^{[1:t]}=\sigma^{[1:t-1]}\cup \left\{g_t(\cdot )\right\}$, respectively. The terms $\lambda_{t}\tilde{v}_t$ and $\lambda_{t-1}\tilde{v}_t$ represent the upper and lower bounds of the cost of committing $\tilde{v}_t$ to obtain the new additional revenue $g_t\left(\tilde{v}_t\right)$ in time slot $t$. Thus the difference between them represents a bound on the ``profit'' obtained in slot $t$. Intuitively, Lemma~\ref{lem:opt_difference_bound} says that one can bound the optimal offline revenue increment by these profit bounds, as shown in \eqref{eq:off_Optimal_value_increment_lower_bound} and \eqref{eq:off_Optimal_value_increment_upper_bound}.

%\textred{The proof of Lemma \ref{lem:opt_difference_bound} is included in the Appendix~\ref{app:proof_opt_difference_bound}.}
%(??? The proof idea is not easy to understand )

\textred{The proof of Lemma \ref{lem:opt_difference_bound} is included in the Appendix~\ref{app:proof_opt_difference_bound}; we give the proof idea here. Compared with the optimal solution under $\sigma^{[1:t-1]}$, the optimal solution under $\sigma^{[1:t]}$ is smaller at $\tau, \forall\tau\leq t-1$, in order to commit $\tilde{v}_t$ to $g_t{(\cdot)}$, which cause a decrement in revenue. Furthermore, the per-unit revenue lost  is upper bounded by $\lambda_t$ and lower bounded by $\lambda_{t-1}$. Combining the two understandings gives the bounds of the increment on optimal revenue at each time.}

%The result follows from observing that the difference between optimal solutions at time $t$ and $t-1$ for each revenue function at time $\tau < t$, $g_\tau$, is concave, and therefore yields a lower bound (respectively upper bound) based on the dual variable $\lambda_t$ (resp. $\lambda_{t-1}$) and the differences in the optimal solution. A further observation that the aggregate optimal inventory used in the later time slot is no smaller than that of the earlier time slot and it takes equal when $\lambda_{t-1}>0$, and the result follows. 
%A similar observation can also be made for the upper bound but instead employing the concavity of $g_t$.   
The upper bound in \eqref{eq:off_Optimal_value_increment_upper_bound} also highlights an intuitive result that the increment of the optimal aggregate revenue from $t-1$  to $t$ is at most $ g_t\left(\hat{v}_t\right)$, i.e., the maximum revenue one can obtain in slot $t$.

% \textred{Lemma \ref{lem:opt_difference_bound} serves as an alternative version of Proposition \ref{prop:sufficient}, and are equivalent in the cases where the inventory is not fully utilized and the optimal offline solution $v_t^*$ is also the maximizer of the function $g_t$, in which case $\lambda_t = 0$. The left hand side of the equations are identical to that of Proposition \ref{prop:sufficient}, and are the difference between the offline optimal of consecutive time periods. $\lambda_t$ and $\lambda_{t-1}$ serves the role of the marginal benefit of inventory at those time periods, and the terms $\lambda_tv_t^*$ and $\lambda_{t-1}v_t^*$ are upper and lower bounds of the cost involved in ``freeing" up inventory to commit to $v_t^*$, while $g_t\left(v_t^*\right)$ represents the marginal revenue. }

% The following proposition, essentially implying a guaranteed ability to pursue the optimal revenue if we had infinite resources, allows us to focus on feasibility based on the sufficiency of the total inventory $\Delta$ in the remainder of the paper. 

% \begin{corollary}
% \label{prop:sufficient}
% We have $\eta_{OPT}\left(\sigma^{[1:t]}\right)-\eta_{OPT}\left(\sigma^{[1:t-1]}\right)\leq g_{t}\left(\hat{v}_{t}\right)$, 
% %$\Delta\left(\tilde{p}\left(t\right)-\tilde{p}\left(t-1\right)\right)\leq g_{t}\left(\hat{v}_{t}\right)$,
% where $\hat{v}_{t}\in[0,\Delta]$ is the maximizer of $g_{t}\left(v\right)$. \label{lem:.step_opt}
% \end{corollary}

%Lemma \ref{lem:opt_difference_bound} provides bounds on the increment of the optimal aggregate revenue in a time slot $t$. 
Our last result in this section, as stated in the lemma below, reveals another subtle yet important property of the increment of the optimal aggregate revenue.

% The following lemma describes how the increment of offline optimal upon a newly revealed $g_\tau$ changes when we interchange the position of $g_\tau$ and $g_{\tau+1}$ under an input sequence.

% is useful in the proof of Lemma \ref{lem:increasing_marginal_worstcase}.
%\txtred{We further show the following lemma, describing that the increment of offline optimal upon a newly revealed $g_t$ decreases when the input sequence is longer, i.e., it contains a larger set of $g$.}
\begin{lem}
\label{lem:interchange_input_opt}
Let $\tilde{\sigma}$ be an input sequence. $\bar{\sigma}$ is another input sequence constructed by interchanging $g_\tau$ and $g_{\tau+1}$ in $\tilde{\sigma}$, for any selected $\tau\in[T]$. We have  
\begin{equation}
\label{eq:interchange_input_opt}
\eta_{OPT}\left(\tilde{\sigma}^{[1:\tau]}\right)-\eta_{OPT}\left(\tilde{\sigma}^{[1:\tau-1]}\right) \geq \eta_{OPT}\left(\bar{\sigma}^{[1:\tau+1]}\right)-\eta_{OPT}\left(\bar{\sigma}^{[1:\tau]}\right).
\end{equation}
\end{lem}
%\txtred{To add: discussion about the lemma above and summation of the proof} The proof is based on Lemma ~\ref{lem:opt_difference_bound}.  

%\textred{which implies that adding any additional functions before the current function cannot increase its contribution to the offline optimal at the time it is introduced.}   

\textred{The left- (resp. right-) hand-side of ~(\ref{eq:interchange_input_opt}) can be regarded as the increment $g_{\tau}$ contributes to  the offline optimal under $\tilde{\sigma}$ (resp. $\bar{\sigma}$). Inequality~(\ref{eq:interchange_input_opt}) means that moving $g_{\tau+1}$ ahead of $g_\tau$ (as under $\bar{\sigma}$) will not increase the contribution of $g_\tau$ to the offline optimal. }
\textred{Lemma \ref{lem:interchange_input_opt} basically states that regardless of the input sequence thus far, the impact or improvement in the offline optimal that $g_\tau$ brings at the time it appears in the input sequence has a ``diminishing effect'' in time. The proof of Lemma ~\ref{lem:interchange_input_opt} is essentially based on the bounds on the increment of offline optimal at each time in  Lemma ~\ref{lem:opt_difference_bound}. We leave the proof in Appendix~\ref{app:proof_interchange_input_opt}.}

\section{\textsf{CR-Pursuit} Algorithmic Framework}
\label{sec:online_alg_framework}

%\textbf{Remark: }our setting corresponds to the case with unknown duration and known $m,M$ in the original one way trading problem \cite{el2001optimal}. This setting is also studied in \cite{yang2017online}.

% However, in practice, when a decision is made, the future revenue functions and stopping time $T$ are not known, and robust decisions have to be made. One consequence is that these decisions are made essentially assuming that there is a possibility that the current time epoch is the stopping time $T$. Thus we are interested in the online setting where we do not assume any further distributional information on the revenue functions and decisions are irrevocable.

%In the following, we first analyze the one-way trading problem to obtain some insights, and then extend our analysis to one-way trading problem with price elasticity.

%\subsection{Online
%Algorithm Design Framework for $GOWT$ problem}
In this section, we present \textsf{CR-Pursuit}, a new algorithmic framework for solving OOIC under the online setting, where the interval length $T$ is not known beforehad and the revenue functions $g_t(\cdot)$, $t\in[T]$, are revealed in a slot-by-slot fashion.

\textsf{CR-Pursuit} is parameterized by a competitive ratio $\pi$, and chooses actions with the goal of ``pursuing'' this competitive ratio, i.e., maintaining the competitive ratio against the offline optimal of the previously observed revenue functions at all time. We derive bounds on the optimal competitive ratios and use them to operate \textsf{CR-Pursuit} accordingly.

%Thus, to apply this framework, it is first necessary to derive competitive bounds and then the bounds are used to design the online algorithm.

In the following, we first present the \textsf{CR-Pursuit} framework and show that one can optimize the only parameter of \textsf{CR-Pursuit} to achieve the best possible competitive ratio, thus significantly reducing the search space of optimal online algorithms. Then, we identify a ``critical'' input sequence that highlights an important structural property of the space of input sequences. By applying \textsf{CR-Pursuit} to this critical sequence, we characterize a lower bound on the optimal competitive ratio as $\ln \theta +1$, where we recall that $\theta=M/m$ is the ratio between the maximum and minimum base prices. Then, for any other input, the performance ratio achieved by \textsf{CR-Pursuit} (with the same parameter) is upper bounded by the product of a problem-dependent factor and the lower bound. This structure not only suggests a principled approach for characterizing the optimal competitive ratio, but also immediately shows that \textsf{CR-Pursuit} (with a parameter being the product of the lower bound and the problem-dependent factor)  achieves the optimal competitive ratio (up to a problem-dependent factor) for solving OOIC among all deterministic algorithms.

% Adam: This paragraph makes the framework seem less novel... The one-way trading problem to some extent have been solved through a similar framework, in that the structure involves finding a worst case input sequence of prices, and sets as a threshold an amount of inventory to sell up to for each incoming price. Together with the generality of concave revenue functions or price elasticity, the worst case input sequence of prices become harder to characterize, e.g., given that eventual trading prices are no longer necessarily above the lower bound $m$. Additionally, the offline optimal is also much more difficult to characterize and compare to, given that it is not always optimal to sell all the inventory at the maximal price shown thus far.

\textbf{\textsf{CR-Pursuit}}. Recall that $\sigma^{[1:t]}=\{g_1,g_2,...,g_t\}$ is the input up to time $t$ and $\eta_{OPT}\left(\sigma^{[1:t]}\right)$ is the corresponding optimal offline revenue. The class of online algorithms that make up the \textsf{CR-Pursuit} framework, denoted as \textsf{CR-Pursuit}($\pi$) and presented in Alg. \ref{alg:Online-algorithm}, can be described as follows: Given any $\pi\ge1$, at the current time $t$,   \textsf{CR-Pursuit}($\pi$) outputs a $\bar{v}_t\in [0, \Delta]$ that satisfies
\begin{equation}
g_t\left(\bar{v}_t\right)=\frac{1}{\pi}\left[\eta_{OPT}\left(\sigma^{[1:t]}\right)-\eta_{OPT}\left(\sigma^{[1:t-1]}\right)\right].\label{eq:keep_cr}
\end{equation}
We remark that such $\bar{v}_t$ always exists, because (i) $g_t(\cdot)$ is a continuous and increasing function and (ii) the right-hand-side of \eqref{eq:keep_cr} is in $\left[g_t(0),g_t\left(\hat{v}_t\right)\right]$ according to Lemma~\ref{lem:opt_difference_bound}.

% \begin{equation}
% \eta^{t}=\eta^{t-1}+g_t\left(v_t\right)=\frac{1}{\pi}\eta_{OPT}\left(\sigma^{[1:t]}\right),\label{eq:keep_cr}
% \end{equation}
% where $\eta^{t}$ and $\eta^{t-1}$ are the revenues of \textsf{CR-Pursuit}$\left(\pi\right)$ up to time $t$ and $t-1$, respectively.

Essentially, \textsf{CR-Pursuit}$\left(\pi\right)$ aims at keeping the \textit{offline-to-online revenue ratio} to be $\pi>1$ at all time, i.e.,
\begin{equation}
    \sum_{\tau=1}^t g_{\tau}\left(\bar{v}_{\tau}\right) = \frac{1}{\pi} \eta_{OPT}\left(\sigma^{[1:t]}\right),\;\;\forall t\in[T]. \label{eq:CR-Pursuit:aggregate.revenue}
\end{equation}

While \textsf{CR-Pursuit}$\left(\pi\right)$ can be defined for any $\pi$, the solution obtained by \textsf{CR-Pursuit}$\left(\pi\right)$ may violate the inventory constraint in OOIC and be infeasible. This motivates the following definition.

\begin{algorithm}[!t]
\caption{\textsf{CR-Pursuit}($\pi$) Online algorithm \label{alg:Online-algorithm}}
\begin{algorithmic}[1]
\STATE \textbf{Input:} $\pi>1$, $\Delta$
\STATE \textbf{Output:} $\bar{v}_t,t\in [T]$
% \STATE Initialize  $\eta^0=0$
\WHILE{$t$ is not the last slot}
% \STATE Compute $\lambda^*$ under input $\sigma^{[1:t]}$ according to Alg. \ref{alg:Offline-algorithm}
\STATE Obtain $\eta_{OPT}\left(\sigma^{[1:t]}\right)$ by solving the convex problem OOIC given the input until $t$, i.e., $\sigma^{[1:t]}$
\STATE Obtain a $\bar{v}_t \in [0,\Delta]$ that satisfies  \eqref{eq:keep_cr}
% \STATE Update $\eta^t=\eta^{t-1}+g_t\left(\bar{v}_t\right)$
\ENDWHILE
\end{algorithmic}
\end{algorithm}

\begin{defn}
\textsf{CR-Pursuit}$\left(\pi\right)$ is feasible if $\Phi_\Delta\left(\pi\right)\leq \Delta$, where
\begin{align}
    \Phi_\Delta\left(\pi\right)\triangleq & \max_{\sigma \in \Sigma}\sum_{t=1}^{T}\bar{v}_t(\sigma),
\end{align}
and $\bar{v}_t(\sigma)$ is the output of \textsf{CR-Pursuit}$\left(\pi\right)$ at time $t$ under the input $\sigma$.
\end{defn}
If \textsf{CR-Pursuit}$\left(\pi\right)$ is feasible, i.e., it can maintain the
offline-to-online revenue ratio to be $\pi$ under all possible input sequences without violating the inventory constraint, then by definition it is $\pi$-competitive. We present a useful observation on $\Phi_\Delta\left(\pi\right)$. % in the following lemma that, given a fixed inventory constraint $\Delta$, among all \textsf{CR-Pursuit} algorithms with feasible competitive ratio, that an algorithm with a (strictly) smaller competitive ratio requires a (strictly) larger output at each time epoch.

%Formally, we have the following result:
\begin{lem}
\label{lem:vt_decreasing}
$\Phi_\Delta\left(\pi\right)$ is strictly decreasing in $\pi$ over $[1,\infty)$.
\end{lem}
 %From Lemma~\ref{lem:vt_decreasing}, it is obvious that $\Phi_\Delta\left(\pi\right)$ is decreasing in $\pi\ge1$.
 Lemma \ref{lem:vt_decreasing} follows naturally since attempting to preserve a smaller competitive ratio requires selling a larger inventory to match the discounted revenue obtained by the offline optimal algorithm. It also implies that if \textsf{CR-Pursuit}$\left(\pi_1\right)$ is feasible for some $\pi_1$, then any online algorithm \textsf{CR-Pursuit}$\left(\pi\right)$ with $\pi\geq \pi_1$ is also feasible. Thus an upper bound on the optimal competitive ratio in this case gives a feasible competitive online algorithm.

% Note that for any initial inventory $\Delta$,  is always feasible, as \textsf{CR-Pursuit}$\left(\theta\right)$. This is because for large enough $\pi$, the right-hand-side in~\eqref{eq:keep_cr} is small and \textsf{CR-Pursuit}$\left(\pi\right)$ only sells a small quanity out of the inventory to maintain the revenue ratio.

% A careful reader may note that the inventory $v_t$ at a time epoch $t$ may not suffice to pursue the competitive ratio, causing another potential infeasibility, i.e., that step $7$ in Algorithm \ref{alg:Online-algorithm} yields no possible solution. Here we present the following proposition to handle this case and allow us to focus on feasibility based on sufficiency of total inventory in the remainder of the paper.

%is straightforward as the marginal revenue is non-decreasing in $t$ and from \left(5\right) we know that given the new input at time $t$, $v_\tau,\tau\in[t-1]$ are non-increasing, thus the optimal revenue for $\tau\in[t-1]$ are non-increasing and the offline optimal at most increases by the maximum value of the new input at time $t$. Also, it is an easy corollary from Lemma ~\ref{lem:opt_difference_bound}. The proof is left in the Appendix.
%In online algorithm design, we are interested in the , defined as $\Phi_\Delta\left(\pi\right)$.

\textbf{The Optimal Competitive Ratio.} We now present a key result, which says that it suffices to focus on \textsf{CR-Pursuit} for achieving the optimal competitive ratio.

\begin{thm}
\label{thm:optimal_cr2}
Let $\pi^*$ be the unique solution to the characteristic equation  $\Phi_\Delta\left(\pi\right)=\Delta$. Then \textsf{CR-Pursuit}($\pi^*$) is feasible and $\pi^*$ is the optimal competitive ratio of deterministic online algorithms.
\end{thm}

% Theorem \ref{thm:optimal_cr2} says that designing an optimal deterministic online algorithm for OOIC can be reduced to solving the  characteristic equation $\Phi_\Delta\left(x\right)=\Delta$.

% By definition, if $\Phi_\Delta\left(\pi\right)\le\Delta$, then \textsf{CR-Pursuit}$\left(\pi\right)$ is feasible and $\pi$-competitive while $\Phi_\Delta\left(\pi\right)>\Delta$ implies that \textsf{CR-Pursuit}$\left(\pi\right)$ may be infeasible under certain input sequences.

% If we can obtain a closed-form expression of $\Phi_\Delta\left(\pi\right)$, then by setting $\Phi_\Delta\left(\pi\right)=\Delta$, we can obtain the minimum competitive ratio $\pi^*$ such that \textsf{CR-Pursuit}$\left(\pi^*\right)$ is feasible and thus it is $\pi^*$-competitive. Indeed, as shown in Theorem \ref{thm:optimal_cr2}, $\pi^*$ is the best competitive ratio among all the deterministic online algorithms.

Before we proceed to prove Theorem \ref{thm:optimal_cr2}, we first present the following lemma characterizing a class of worst case inputs for \textsf{CR-Pursuit}. The results will be used in the proof of Theorem~\ref{thm:optimal_cr2}.

\begin{lem}
\label{lem:increasing_marginal_worstcase}
For any \textsf{CR-Pursuit}($\pi$), there exists an input sequence $\sigma$ such that (i) \textsf{CR-Pursuit}($\pi$) sells exactly the $\Phi_{\Delta}\left(\pi\right)$ amount of inventory and (ii) $g_t'\left(\bar{v}_t(\sigma)\right)$ is non-decreasing in $t$.
\end{lem}

\textred{Lemma  \ref{lem:increasing_marginal_worstcase} states that to compute $\Phi_{\Delta}\left(\pi\right)$, it is sufficient to focus on the input sequences that will lead to a non-increasing sequence of marginal revenue $g_t'\left(\bar{v}_t\right)$ at the solution obtained by \textsf{CR-Pursuit}($\pi$).
Intuitively, these sequences will cause the \textsf{CR-Pursuit}($\pi$) algorithm to sell large quantities at lower prices in the early slots, without knowing that the marginal revenues at later slots are higher, which is exploited by the offline optimal solution. As a result, \textsf{CR-Pursuit}($\pi$) will need to sell the ``worst'' amount of inventory to keep the revenue ratio $\pi$.

The proof of Lemma \ref{lem:increasing_marginal_worstcase} is provided in Appendix~\ref{app:proof_increasing_marginal_worstcase}, based on the subtle yet important property of the offline optimal aggregate revenue in Lemma~\ref{lem:interchange_input_opt}. The idea of the proof, roughly speaking, is that if the worst case input sequence is not as stated, then we can swap revenue functions within the sequence to construct a new worst case one that satisfies the conditions.} Putting the preceding lemmas together, we are now ready to prove Theorem \ref{thm:optimal_cr2}.

\begin{proof}[Proof of Theorem \ref{thm:optimal_cr2}]
The feasibility of \textsf{CR-Pursuit}$\left(\pi^{*}\right)$ is because of the definition of $\pi^*$. What remains to be proved is that $\pi^{*}$ is the optimal competitive ratio.

Consider an arbitrary deterministic online algorithm different from \textsf{CR-Pursuit}$\left(\pi^{*}\right)$, denoted as $\mathcal{A}$. We will show that $\mathcal{A}$ cannot achieve an offline-to-online revenue ratio smaller than $\pi^{*}$ over an input sequence that we construct.

Let $\tilde{\sigma}^{[1:T]}=\{\tilde{g}_1,\tilde{g}_2,...,\tilde{g}_T\}$ be a worst case input sequence of \textsf{CR-Pursuit}$\left(\pi^{*}\right)$ that satisfies the conditions in Lemma~\ref{lem:increasing_marginal_worstcase}. Let $\bar{v}_t$ and $v^{\mathcal{A}}_t$ be the corresponding solutions of \textsf{CR-Pursuit}$\left(\pi^{*}\right)$ and $\mathcal{A}$ at time $t$, respectively. We have
\begin{itemize}
    \item $\sum_{t=1}^{T}\bar{v}_t=\Phi_\Delta\left(\pi^*\right)=\Delta$;
    \item $\tilde{g}_t'\left(\bar{v}_t\right)$ is non-decreasing in $t$.
\end{itemize}

We now construct an input sequence over which $\mathcal{A}$ cannot achieve an offline-to-online revenue ratio smaller than $\pi^{*}$, by feeding $\tilde{g}_1,\tilde{g}_2,...,\tilde{g}_T$ to $\mathcal{A}$ and stop at any time that we need.

We first present $\tilde{g}_1$ to $\mathcal{A}$ in the first slot. If
$v^{\mathcal{A}}_1\le \bar{v}_1$, we stop and set $T=1$ in this constructed sequence. In this case, we have
\[
\tilde{g}_1\left(v^{\mathcal{A}}_1\right)\le \tilde{g}_1\left(\bar{v}_1\right)=\frac{1}{\pi^*}\eta_{OPT}\left(\tilde{\sigma}^{[1:1]}\right),
\]
thus the competitive ratio of $\mathcal{A}$ is at least $\pi^{*}$. Otherwise we have $v^{\mathcal{A}}_1>\bar{v}_1$ and we continue to present $\tilde{g}_2$ to $\mathcal{A}$ in the second slot.

In general, if at time $t$ the total selling quantity of $\mathcal{A}$ so far is no larger
than that of \textsf{CR-Pursuit}$\left(\pi^{*}\right)$, i.e., $\sum_{\tau=1}^{t}\bar{v}_\tau$, we end the trading period. Otherwise, we continue to the $t+1$ slot and present $\mathcal{A}$ with the revenue function $\tilde{g}_{t+1}(\cdot)$.

Let $\tau$ be the earliest slot such that at the end of time $\tau$, the total selling quantity of $\mathcal{A}$ is less than that of \textsf{CR-Pursuit}$\left(\pi^{*}\right)$. Such $\tau$ exists; otherwise, we will have $\sum_{t=1}^{T}v^{\mathcal{A}}_t>\sum_{t=1}^{T}\bar{v}_t=\Delta$, which implies that $\mathcal{A}$ is not feasible. Given such $\tau\in[T]$, we have
\begin{align}
    \sum_{\xi=1}^{t}v^{\mathcal{A}}_\xi  & > \sum_{\xi=1}^{t}\bar{v}_\xi, \forall t\in [\tau-1], \label{eq:quantity.relationship.before.tau}\\
    \mbox{and }\;\;\;\sum_{\xi=1}^{\tau}v^{\mathcal{A}}_\xi & \le\sum_{\xi=1}^{\tau}\bar{v}_\xi.\label{eq:quantity.relationship.at.tau}
\end{align}

We now show that, for the input sequence $\tilde{\sigma}^{[1:\tau]}$, the aggregate revenue of $\mathcal{A}$ is no larger than that of \textsf{CR-Pursuit}$\left(\pi^*\right)$, i.e.,
\begin{equation}
\sum_{\xi=1}^{\tau}\tilde{g}_\xi\left(v^{\mathcal{A}}_\xi\right) -  \sum_{\xi=1}^{\tau}{\tilde{g}_\xi\left(\bar{v}_\xi\right)}\leq 0, \label{eq:revenue_A_leq_revenue_CR-Pursuit}
\end{equation}
which then implies that the online algorithm $\mathcal{A}$ is at best $\pi^*$-competitive. By the concavity of $\tilde{g}_t(\cdot)$, we have
\begin{equation*}
\hspace*{-.55\columnwidth}  \sum_{\xi=1}^{\tau}\left[\tilde{g}_\xi\left(v^{\mathcal{A}}_\xi\right)-\tilde{g}_\xi\left(\bar{v}_\xi\right)\right]
\end{equation*}
\vspace*{-2mm}
\begin{align*}
  \hspace*{.15\columnwidth} \leq\;\; &  \sum_{\xi=1}^{\tau}\tilde{g}'_\xi\left(\bar{v}_\xi\right)\left(v^{\mathcal{A}}_\xi-\bar{v}_\xi\right) \\
    =\;\; &  \tilde{g}'_\tau\left(\bar{v}_\tau\right)\left(\sum_{\xi=1}^{\tau}v^{\mathcal{A}}_\xi-\sum_{\xi=1}^{\tau}\bar{v}_\xi\right) \\
     & - \sum_{t=1}^{\tau-1} \left[\tilde{g}'_{t+1}\left(\bar{v}_{t+1}\right) - \tilde{g}'_{t}\left(\bar{v}_{t}\right)\right] \left(\sum_{\xi=1}^{t}v^{\mathcal{A}}_\xi-\sum_{\xi=1}^{t}\bar{v}_\xi\right).
\end{align*}
By \eqref{eq:quantity.relationship.at.tau} and that $\tilde{g}'_\tau\left(\bar{v}_\tau\right)\geq 0$ as $\tilde{g}_\tau(\cdot)$ is an increasing function, the first term in the last line of derivation is non-positive. By \eqref{eq:quantity.relationship.before.tau} and that $\tilde{g}_t'\left(\bar{v}_t\right)$ is non-decreasing in $t$, each term in the summation in the last line of derivation is non-negative. As such, the right-hand-side is non-positive and the inequality in \eqref{eq:revenue_A_leq_revenue_CR-Pursuit} holds.
\end{proof}
% \txtred{Need to show: Since for any $v$, $\tilde{g}_{\tau}\left(v\right)$ are increasing in $\tau$\left(Intuition: it is better to sell more in the later slots\right)}, $\mathcal{A}$
% would have achieved a larger revenue by selling
% exactly $v\left(\tau\right)$ for any $\tau\in[t'-1]$ and by selling $v_{\mathcal{A}}^{*}\left(t'\right)=v_{\mathcal{A}}\left(t'\right)+\sum_{\tau=1}^{t'-1}v_{\mathcal{A}}\left(\tau\right)-\sum_{\tau=1}^{t'-1}v\left(\tau\right)$
% at time $t'$. Namely, we have
% \[
% \eta_{\mathcal{A}}^{t'}\le\sum_{\tau=1}^{t'-1}\tilde{g}_{\tau}\left(v\left(\tau\right)\right)+\tilde{g}_{t'}\left(v_{\mathcal{A}}^{*}\left(t'\right)\right),
% \]
% where $\eta_{\mathcal{A}}^{t'}$ is the revenue
% of $\mathcal{A}$ up to time $t'$. However, from (\ref{eq:t'slot}),
% we know $v_{\mathcal{A}}^{*}\left(t'\right)\le v\left(t'\right)$ and thus we have
% \[
% \eta_{\mathcal{A}}^{t'}\le\sum_{\tau=1}^{t'}\tilde{g}_{\tau}\left(v\left(\tau\right)\right)=\eta_{OPT}\left(\tilde{\sigma}^{[1:t']}\right)/\pi^{*}.
% \]

%ensuring that such a competitive ratio \left(and corresponding online algorithm\right) can indeed be found.
\section{Competitive Analysis of \textsf{CR-Pursuit}\label{UB_CR}}

The results in the previous section highlight a principled approach to construct an optimal online algorithm. Specifically, the first step is to mathematically characterize $\Phi_\Delta\left(\pi\right)$. Then we solve the characteristic equation $\Phi_\Delta\left(\pi\right)=\Delta$ to obtain the optimal competitive ratio $\pi^*$, and  \textsf{CR-Pursuit}($\pi^*$) is an optimal online algorithm for solving OOIC. For special cases such as the one-way trading problem \cite{el2001optimal} where $g_t\left(v\right)=p(t)\cdot v$, we can obtain the closed-form expression of $\Phi_\Delta\left(\pi\right)$ and compute the optimal competitive ratio (as demonstrated in Sec.~\ref{OWT}). However, it is difficult to obtain a closed-form expression for general concave revenue functions. Instead, we characterize an upper bound on $\Phi_\Delta\left(\pi\right)$, based on which we can give an upper bound on the optimal competitive ratio $\pi^*$ and consequently a feasible online algorithm.

Before moving to the upper bound though, it is helpful to understand a lower bound on the optimal competitive ratio.  For this, we can simply refer to the literature on one-way trading.  In particular, it has been shown that the optimal competitive ratio of the classic one-way trading problem is $\ln \theta +1$~\cite{el2001optimal,yang2017online}.  Since OOIC covers one-way trading as a special case, the optimal competitive ratio for any online algorithm solving OOIC is lower bounded by $\ln\theta+1$. Interestingly, it is possible to interpret this bound in the context of the \textsf{CR-Pursuit} framework.  In particular, in Sec.~\ref{OWT}.  we identify the worst case input in one-way trading (defined in Sec.~\ref{OWT}) as a ``critical" input sequence, reflecting an interesting structure on the space of input sequences. By applying \textsf{CR-Pursuit} to this sequence, we characterize a lower bound on the optimal competitive ratio as $\ln\theta+1$. %This structure suggests a principled approach to characterize the optimal competitive ratio.

%The critical input also plays an important role in establishing the upper bound for $\Phi_\Delta\left(\pi\right)$.
%moving to the upper bound, recall that each competitive ratio bound used leads to a different version of \textsf{CR-Pursuit}.  Too small a competitive ratio may lead to an infeasible algorithm. ,
It turns out that for any other inputs, the performance ratio achieved by \textsf{CR-Pursuit} is upper bounded by the product of a problem-dependent factor and the lower bound $\ln\theta+1$. This insight leads to the following results.
% that show that \textsf{CR-Pursuit} can achieve an order-optimal competitive ratio.
%based on parameters under different classes of revenue functions.
\begin{thm}
\label{thm:UB_OOIC} Recall that $\mathcal{G}$ is the set of all possible $g(\cdot)$ and $\hat{v}\in[0,\Delta]$ is the maximizer of $g(\cdot)$. Let $c=\sup_{g\in\mathcal{G}}\frac{g'\left(0\right)}{g\left(\hat{v}\right)/\hat{v}}$, then
the optimal competitive ratio $\pi^*$ satisfies
\[
    \ln\theta +1 \leq \pi^* \leq c\left(\ln\theta +1\right).
\]
\end{thm}

\textred{Theorem~\ref{thm:UB_OOIC} characterizes an upper bound on the optimal competitive ratio in the case for general revenue functions $g_t$, and also implies that \textsf{CR-Pursuit}$\left(c\left(\ln \theta +1\right)\right)$ is feasible and its competitive ratio is $c\left(\ln \theta +1\right)$.} Note that $c$ is a constant that depends on the gradient properties (in particular the base price) and the maximizers of the revenue functions\footnote{While $c$ is a constant when the family of revenue functions are fixed, it is indeed true that $c$ could presumably be driven to be infinitely large, e.g., with revenue functions that are concave and increasing. This parameter $c$ can be seen in an economical sense as a comparison between the base price and the average price at the maximizer of the function. Since the former is already bounded in $[m,M]$, we look at the case when the latter is small. These situations are hard to derive any interesting online optimization as the functions require too much commitment even in bad time epochs, and have low average prices. This results in low committed average prices while the offline optimal may eventually not have to participate in these time epochs. }. For many interesting problems, this $c$ is bounded and small. For example, for the one-way trading problem where the revenue functions are linear, i.e., $g_t\left(v\right)=p\left(t\right) v,\forall t\in [T]$, we have $c=1$. As another example, for the one-way trading with linear price elasticity where the revenue functions are quadratic, i.e., $g_t\left(v\right)=\left(p\left(t\right)-\alpha_tv\right)v,\forall t\in [T]$, we have $c=2$.

% For ease of discussions, we define
% \begin{align}
%     \tilde{p}\left(t\right)\triangleq \frac{1}{\Delta}\eta_{OPT}\left(\sigma^{[1:t]}\right), \forall t\in [T],
% \end{align}
% which can be interpreted as the averaged trading price for the offline optimal algorithm given the input $\sigma^{[1:t]}$.

%\textred{Lemma \ref{lem:upper_bound_v_t_OOIC} limits the amount of inventory required to pursue the competitive ratio, in terms of the problem-dependent factor $c$, the competitive ratio $\pi$, the inventory size $\Delta$, and the difference between the consecutive averaged trading prices of the offline optimal solution. Lemma \ref{lem:bound_on_x_OOIC} bounds the differences between consecutive average trading prices. That helps in the proof of Lemma \ref{lem:to_goal_OOIC}, which gives an upper bound relying only on $\theta$ and not the actual sequence of functions or prices. }
To prove this theorem, we use a sequence of lemmas elaborated as follows.
\textred{We begin with Lemma~\ref{lem:upper_bound_v_t_OOIC}, which gives an upper bound on the total selling quantity by \textsf{CR-Pursuit}$\left(\pi\right)$ in each time slot to maintain the offline-to-online revenue ratio.
Recall that the output of the algorithm \textsf{CR-Pursuit}$\left(\pi\right)$ at slot $t$, $g_t(\bar{v}_t)$ is given in \eqref{eq:keep_cr}, and $p\left(t\right)=g_{t}'\left(0\right)$ is the base price at slot $t$.}
% and~\eqref{eq:eta_t}
% ,  is $\bar{v}_{t}$ that satisfies
% \begin{equation}
% g_{t}\left(\bar{v}_{t}\right)=\frac{1}{\pi}\left[\eta_{OPT}\left(\sigma^{[1:t]}\right) - \eta_{OPT}\left(\sigma^{[1:t-1]}\right)\right].\label{eq:keep_ratio_OOIC_UB}
% \end{equation}

\begin{lem}
\label{lem:upper_bound_v_t_OOIC}
For any input sequence $\sigma$, we have
$$\bar{v}_{t}\leq c \frac{g_t(\bar{v}_t)}{p(t)},\forall t\in [T].$$
\end{lem}
\textred{The proof of Lemma~\ref{lem:upper_bound_v_t_OOIC} is included in Appendix~\ref{apx:proof.lem.upper_bound_v_t_OOIC}, by leveraging the definition of $c$ and that $g_t\left(v\right)$  is an increasing concave function.

Next, we present an interesting result that bounds the contribution to the online revenue in all the slots whose base prices is no higher than any specific threshold. }

\begin{lem}
\label{lem:bound_on_x_OOIC}
For any input sequence $\sigma\in \Sigma$, for any threshold price $p\in[m,M]$, we have
\[
\sum_{\left\{t:\;p\left(t\right)\leq p\right\}}g_t(\bar{v}_t) \leq \frac{1}{\pi} p \cdot \Delta.
\]
\end{lem}
\textred{Lemma ~\ref{lem:bound_on_x_OOIC} is intuitive in that the left-hand-side is the online revenue obtained by \textsf{CR-Pursuit}($\pi$) in the slots whose base prices is not higher than $p$. The right-hand-side is simply the maximum revenue achievable by \textsf{CR-Pursuit}($\pi$) in these slots according to its design. In the proof, we first observe that if  $p\left(t\right)<p$, $\forall t\in [T]$, the result is immediate. As for general cases, based on Lemma ~\ref{lem:interchange_input_opt}, we can construct new input sequences by moving forward the slots with $p\left(t\right)\leq p$ in $\sigma$, while increasing the online revenue in the slots that we are interested in. At last, we obtain an input sequence with larger online revenue in these slots, which are now all in the beginning of the input sequence. The total online revenue in them is bounded by $p\cdot \Delta / \pi$. Lemma \ref{lem:bound_on_x_OOIC} allows to prove a key step used in the proof of Theorem~\ref{thm:UB_OOIC} below. }

\begin{lem}
\label{lem:to_goal_OOIC}For any input sequence $\sigma$, we have
\begin{equation}
\sum_{t=1}^{T}\frac{g_t(\bar{v}_t)}{p\left(t\right)}\leq \frac{\Delta}{\pi}\left(\ln\theta+1\right).\label{eq:sum_bound}
\end{equation}
\end{lem}
\textred{The idea to prove Lemma ~\ref{lem:to_goal_OOIC} is to construct an optimization problem, whose optimal objective value bounds the left-hand-side in \eqref{eq:sum_bound}, subject to the constraint from Lemma ~\ref{lem:bound_on_x_OOIC}. Then we show the optimal objective value can be further upper bounded by the right-hand-side in \eqref{eq:sum_bound}.}

We are now ready to prove Theorem \ref{thm:UB_OOIC}.

%\textbf{SHOULD THIS BE IN A PROOF ENVIRONMENT OR DO YOU MEAN THIS AS A SKETCH?}
\begin{proof}[Proof of Theorem \ref{thm:UB_OOIC}]
It is clear that \textsf{CR-Pursuit} is at best $\left(\ln\theta+1\right)$-competitive, as it covers the one-way trading problem as a special case, which has an optimal competitive ratio of $\ln\theta +1$.

To establish the upper bound, by Lemmas~\ref{lem:upper_bound_v_t_OOIC} and \ref{lem:to_goal_OOIC}, we observe
\begin{align*}
\Phi_\Delta\left(\pi\right)&=\max_{\sigma\in \Sigma}\sum_{t=1}^{T}\bar{v}_t \le  \sum_{t=1}^{T} c \frac{g_t(\bar{v}(t)}{p\left(t\right)}\\
& \le c \frac{\Delta}{\pi} \left(\ln \theta +1\right).
\end{align*}
By solving $c \frac{\Delta}{\pi} \left(\ln \theta +1\right)=\Delta$, we get that $\bar{\pi}=c\left(\ln \theta +1\right)$ and $\Phi_\Delta\left(\bar{\pi}\right)\le \Delta$. Then according to the definitions, \textsf{CR-Pursuit}$\left(\bar{\pi}\right)$ is feasible and is $\bar{\pi}$-competitive. Hence, $\bar{\pi}$ is an upper bound for the optimal competitive ratio $\pi^*$.
\end{proof}
Theorem \ref{thm:UB_OOIC} implies that \textsf{CR-Pursuit} achieves the optimal competitive ratio (up to a problem-dependent factor $c$) among all deterministic online algorithms.

% \subsection{Online Algorithm Design Framework}

% Generalize to the case with objective function being general concave
% function?\hanling{Qiulin: pls add some discussion here}

% We summarized the online algorithm design framework for one-way trading problem as the following:
% \begin{itemize}
%     \item Step 1: define the class of online algorithm \textsf{CR-Pursuit}$\left(\pi\right)$: at each slot, \textsf{CR-Pursuit}$\left(\pi\right)$ try to keep the offline-to-online revenue ratio to be $\pi$.

%     \item Step 2: define $\Phi_\Delta\left(\pi\right)$ as the maximum possible inventory \textsf{CR-Pursuit}$\left(\pi\right)$ needed to keep the revenue ratio to be $\pi$.

%     \item Step 3: find the minimum $\pi$ such that $\Phi_\Delta\left(\pi\right)$ is upper
% bounded by the inventory. This can make sure that the inventory constraint will never be violated and the competitive ratio of \textsf{CR-Pursuit}$\left(\pi\right)$ is $\pi$.
% \end{itemize}
%  Step 3 is the most difficult step. For simple setting, one can obtain the close-form expression of $\Phi_\Delta\left(\pi\right)$, as we did in Sec.~\ref{subsec:Re-invent-Optimal-OWT}. In this case, the minimum $\pi$ can be found by analyzing function $\Phi_\Delta\left(\pi\right)$ when $\pi\in [1,\infty)$. For complex setting, for example the one-way trading with price elasticity considered in Sec.~\ref{OWTPE}, one can try to compute an upper bound of $\Phi_\Delta\left(\pi\right)$.

%\input{concave.tex}

\section{Application to One-way Trading}
\label{application}

In this section, we apply  \textsf{CR-Pursuit}  to the classic one-way trading problem \cite{el2001optimal} and its generalizations, illustrating that the framework can both match state-of-the-art results for the classic setting and provide new results for generalizations that have previously resisted analysis. In particular, using the CR-Pursuit framework, we obtain an online algorithm matching the optimal competitive ratio  $(\ln \theta +1 )$ for the classic one-way trading problem in Proposition \ref{thm:owt-alg} and a near-optimal $(\ln \theta + 4/3)$ result for the case with linear price elasticity in Theorem \ref{thm:owt-pe}. Furthermore, the algorithmic framework also extends to any convex price elasticity, and yield online algorithms with order-optimal competitive ratio in these cases.  

This section also provides an illustration of how the framework can be applied to specific problem domains to obtain tighter competitive ratio upper bounds that the generic ones under general settings. In particular, 
% bounds may be obtained in general given gradient properties at the origin and optimal solutions of revenue functions, but tighter bounds can be obtained given a more specific family of revenue functions, e.g., 
for one-way trading with linear price elasticity, the upper bound derived from Sec. \ref{UB_CR} is $2(\ln \theta + 1)$ while the bound obtained in this section is $\ln \theta + 4/3$.

In Sec. \ref{OWT}, we obtain a close-form expression of $\Phi_\Delta(\pi)$ and compute the optimal $\pi^*$ in this special case. In Sec. \ref{sec:OWT_PE}, we show the ease of generalizing the one-way trading problem, to cases where price formation include price elasticity, an aspect that has been left out and desired in the one-way trading community. 

\subsection{Classic One-way Trading}
\label{OWT}
In the classic one-way trading problem, a trader owns some assets (e.g., one dollar) at the beginning and aims to exchange it into another assets (e.g., yen) as much as possible,
depending on the price (e.g., exchange rate). Thus, the one-way trading problem is a special case of the OOIC problem with $g_t(v_t)=p(t)v_t$ for all $t\in[T]$ and the input at time $t$  can be simplified as $p(t)$. 

As a direct application, one can obtain from Sec. \ref{UB_CR} that the upper bound for the one-way trading problem is $\ln \theta+1$, which matches the lower bound. Thus, we immediately know that the optimal competitive ratio for one-way trading is $\ln \theta+1$ and \textsf{CR-Pursuit}$(\ln\theta+1)$ is an optimal deterministic online algorithm. In this section, with the aim of demonstrating the possibility of mathematically characterizing $\Phi_\Delta(\pi)$ in specific problems, we first derive a closed-form expression of $\Phi_\Delta(\pi)$, then we obtain the optimal competitive ratio $\pi^*$ by solving the characteristic equation $\Phi_\Delta(\pi)=\Delta$.

In the classic one-way trading problem, given any input up to time $t$, denoted as $\sigma^{[1:t]}\triangleq\{p(1),p(2),...,p(t)\}$, the optimal offline revenue can be expressed as
$\eta_{OPT}(\sigma^{[1:t]})=\Delta\cdot\max\sigma^{[1:t]}.$
Given any $\pi\ge1$, we focus on  \textsf{CR-Pursuit}$(\pi)$ defined in Sec. \ref{sec:online_alg_framework}. At time $t$, \textsf{CR-Pursuit}$(\pi)$ sells the amount $\bar{v}_t\in[0,\Delta]$ that satisfies:
\begin{equation}
\bar{v}_t=\frac{1}{\pi \cdot p(t)}\left[\eta_{OPT}\left(\sigma^{[1:t]}\right)-\eta_{OPT}\left(\sigma^{[1:t-1]}\right)\right].
\end{equation}
As discussed, \textsf{CR-Pursuit}$\left(\pi\right)$ aims at keeping the \textit{offline-to-online revenue ratio} to be $\pi>1$ at all time.

% \begin{equation}
% \frac{\eta_{OPT}\left(\sigma^{[1:t]}\right)}{\pi}=\eta^{t-1}+p(t)v_t,\label{eq:keep_cr_OWT}
% \end{equation}
% where $\eta^{t-1}$ is the revenue of the online algorithm \textsf{CR-Pursuit}$(\pi)$ up to time $t-1$. Clearly, we have $\eta^{0}=0$ and
% \begin{equation}
% \eta^{t}=\eta^{t-1}+v_tp(t).\label{eq:eta_t_OWT}
% \end{equation}
% Essentially, (\ref{eq:keep_cr_OWT}) and (\ref{eq:eta_t_OWT}) imply that the online algorithm \textsf{CR-Pursuit}$(\pi)$ tries to keep the offline-to-online
% revenue ratio at each slot to be $\pi$, i.e., we have $\forall t\in[T]$,
% $\eta_{OPT}(\sigma^{[1:t]})/\eta^{t}=\pi.$

From Sec. \ref{sec:online_alg_framework}, we know that if $\Phi_\Delta(\pi)\le\Delta$, then \textsf{CR-Pursuit}$(\pi)$ is feasible and it is $\pi$-competitive. In the following, our goal is to derive a close-form expression of $\Phi_\Delta(\pi)$.

Observe that at slot $t$, the selling decision of $\textsf{CR-Pursuit}({\pi^*})$
can be simplified as 
\begin{align}
\bar{v}_t & =\frac{\Delta}{\pi^* \cdot p(t)}\left(\max\sigma^{[1:t]}-\max\sigma^{[1:t-1]}\right).\label{eq:keep_cr_OWT}
\end{align}
This suggests that $\textsf{CR-Pursuit}({\pi^*})$ will sell only when the
current price is higher than the best price so far. With this observation, we have the following lemma.

\begin{lem}
\label{lem:increasing_price}For \textsf{CR-Pursuit}$(\pi)$ with $\pi\ge1$,
given any input $\sigma^{[1:T]}$, to compute $\Phi_\Delta(\pi)$,
it is sufficient to consider increasing-price sequences.
\end{lem}

\textred{Lemma \ref{lem:increasing_price} is a corollary of Lemma \ref{lem:increasing_marginal_worstcase} in that while the marginal prices are determined by the participation of the algorithm in the latter, it is constant here in the classic one-way trading problem. }Lemma \ref{lem:increasing_price} can be proved by observing that the revenue of both the offline and online algorithms remain unchanged if the current price is not the highest price so far, and in that case removing this price from the input sequence will not affect the behaviors of both the offline and online algorithms.
From Lemma \ref{lem:increasing_price}, we know that it is sufficient to consider the following increasing price sequence with length $n\le T$:
\begin{equation}
m\leq p_1<p_2<\cdots <p_n\leq M.\label{critical_seq}
\end{equation}
Under the given price sequence, the optimal offline revenue at time $t\in[n]$ can be simplified as 
$$\eta_{OPT}\left(\sigma^{[1:t]}\right)=p_t\Delta.$$
According to (\ref{eq:keep_cr_OWT}), the output of \textsf{CR-Pursuit}$(\pi)$ is given by
$$
\bar{v}_t=\frac{\Delta}{\pi}\frac{p_{t}-p_{t-1}}{p_{t}},\forall t\in [n],
$$
where $p_0=0$.
Then we have 
\begin{align*}
\Phi_\Delta({\pi}) &
= \max_{p_{1},p_{2},\cdots,p_{n}}\sum_{t=1}^{n}\bar{v}_t \\
& =\max_{p_{1},p_{2},\cdots,p_{n}}\frac{\Delta}{{\pi}}\left(1+\frac{p_{2}-p_{1}}{p_{2}}+\cdots+\frac{p_{n}-p_{n-1}}{p_{n}}\right)\\
 & \stackrel{(a)}{=}\frac{\Delta}{{\pi}}\left(1+\int_{m}^{M}\frac{1}{x}dx\right)=\frac{\Delta}{{\pi}}\left(1+\ln\theta\right),
\end{align*}
where (a) holds when the input sequence in \eqref{critical_seq} satisfies  $n\to\infty$ and $p_i\to p_{i+1},\forall i\in [n-1]$. Indeed, this is the worst case input sequence for one-way trading problem, also known as the ``critical'' input sequence.
Setting $\Phi_\Delta({\pi})=\Delta$ yields
the solution that ${\pi^*}=\ln\theta+1$. Consequently, we have
the following result.
\begin{prop}\label{thm:owt-alg}
With  ${\pi^*}=\ln\theta+1$, $\textsf{CR-Pursuit}({\pi^*})$ is feasible and an optimal online algorithm for the one-way trading problem.
\end{prop}
The proof follows the same idea and approach as that of Theorem~\ref{thm:optimal_cr2}. We leave it as an exercise for readers.
% The analysis in this section shows that one can apply 

% Using a similar technique as used in Sec. \ref{subsec:online_alg_framework} and in  \cite{el2001optimal}, we can show that ${\pi^*}$ is the optimal competitive ratio for one-way trading problem. We present the result in the following for completeness.

% \begin{thm}
% \label{thm:optimal_pi}Any deterministic online algorithm for one-way
% trading problem has a competitive ratio that is no smaller than ${\pi^*}$. 
% \end{thm}

\subsection{One-way Trading with Price Elasticity}
\label{sec:OWT_PE}

In this subsection, we consider the one-way trading problem in a  generalized setting with an additional flexibility on the price model playing the role of \emph{price elasticity}. We assume that price is affected by the total quantity sold at each slot, implying that the decision of how much to sell affects the trading price, usually known in the economics literature as \emph{price elasticity}. 

%Elastic goods are typically ones with have many substitutes, i.e., the market have a lot of other choices to choose from. Typical examples include motor vehicles, furniture, professional services, while examples of inelastic goods that are needed regardless of pricing are gas, electricity, water and food. In larger scale, elastic goods also include company shares or other commodities such as Bitcoin, which present large decrements when large volumes are sold at one particular time.

Specifically, we assume that at each slot $t\in[T]$, the price elasticity, defined as $\triangleq f_{t}(v)$,
is a convex non-negative function of the selling quantity with $f(0)=0$. Under this setting, the revenue function at time $t$ becomes $g_t(v)=(p(t)-f_{t}(v))v.$ 
This setting can be considered as a special case of OOIC and the input at time $t$  can be simplified as $\left(p(t),f_t(v)\right)$. Here we have $g'_t(0)=p(t)\in[m,M]$ and  $f_t(v)\in [0,+\infty),\forall v\in[0,\Delta], f_t(0)=0$. Namely, the set of all possible revenue functions can be expressed as
\begin{align*}
\mathcal{G}=&\left\{g_t(v)|g_t(v)=(p(t)-f_t(v))v,p(t)\in [m,M],\right.\\ 
&\left.f_t(v)\in [0,+\infty),\forall v\in[0,\Delta], f_t(0)=0\right\}.
\end{align*} 
Note that when $f_t(v)=0,\forall t\in[T]$, the problem reduces to one-way trading problem considered in Sec.~\ref{OWT}. Thus we note that any deterministic online algorithm in one-way trading with price elasticity has a competitive ratio of at least $\ln \theta +1$. When $f_t(v)=\alpha_t v,\alpha_t\ge 0, \forall t\in [T]$, the problem becomes a one-way trading problem with linear price elasticity, which is a common setting in economic literature, e.g., in Cournot competition \cite{pang2017efficiency}.

Consider the online algorithm \textsf{CR-Pursuit}$(\pi)$ defined in Sec.~\ref{sec:online_alg_framework}. When there is price elasticity in the setting, it is difficult to obtain the closed-form expression of $\Phi_\Delta(\pi)$. We follow the analysis in Sec. \ref{UB_CR} to obtain an upper bound on $\Phi_\Delta(\pi)$. 

% Recall that $\tilde{p}(t)\triangleq\frac{\eta_{OPT}(\sigma^{[1:t]})}{\Delta}$, and we have  $\tilde{p}(t)\in[0,M]$ and $\tilde{p}(t)$ is non-decreasing in $t$. The output
% of the algorithm \textsf{CR-Pursuit}$(\pi)$ at time $t$ satisfies
% \[
% \frac{\Delta(\tilde{p}(t)-\tilde{p}(t-1))}{\pi}=v_{t}(p(t)-f_{t}(v_{t})).
% \]
\textred{In particular, restating Lemma~\ref{lem:upper_bound_v_t_OOIC} under the parametric assumptions of $g_t(v_t)$ in the one-way trading problem with price elasticity, we can upper bound the selling quantity of \textsf{CR-Pursuit}$(\pi)$ at each slot with a better characterization of $c$, reflected in the following lemma.}
\begin{lem}
\label{upper_bound_v_t_OWTPE}
For any input sequence $\sigma$, we have $$\bar{v}_{t}\leq c\frac{g_t(\bar{v}_t)}{p(t)},\forall t\in [T],$$ where  
\begin{equation}
    c=2\left(1+\sqrt{1-{1}/{\pi}}\right)^{-1}. \label{eq:c.OWT.with.price.elasity}
\end{equation}
\end{lem}

We note that value of $c$ given in \eqref{eq:c.OWT.with.price.elasity} is smaller than that derived in Lemma~\ref{lem:upper_bound_v_t_OOIC}. The idea of the proof in Appendix~\ref{app:proof_upper_bound_v_t_OWTPE} is similar to that of Lemma ~\ref{lem:upper_bound_v_t_OOIC}, but we further utilize the special structure of $g_t(\cdot)$ here (i.e., the convexity of $f_t(\cdot)$). The tighter characterization of $c$ allows us to develop an online algorithm with better competitive ratio as compared to the one obtained as a result of Sec.~\ref{UB_CR}.

%We can have the following bound for $\Phi_\Delta(\pi)$. 
\begin{lem}
\label{lem:Phi_pi}For \textsf{CR-Pursuit}$(\pi)$ with $\pi\ge1$, we have $$\Phi_\Delta(\pi)\le \bar{\Phi}(\pi),$$
where $\bar{\Phi}(\pi)\triangleq 2\Delta\left[\pi\left(1+\sqrt{1-{1}/{\pi}}\right)\right]^{-1}(\ln\theta+1).$
\end{lem}

Lemma \ref{lem:Phi_pi} shows that $\Phi_\Delta(\pi)$ is upper bounded
by $\bar{\Phi}(\pi)$. It is easy to show that $\bar{\Phi}(\pi)$ is decreasing in $\pi\ge 1$.
Thus by setting $\bar{\Phi}(\bar{\pi})=\Delta$, we can guarantee that \textsf{CR-Pursuit}$(\bar{\pi})$ is feasible. Then we have the following result, which shows that the competitive ratio of \textsf{CR-Pursuit}$(\bar{\pi})$ is $\ln\theta+\Omega(1)$.

\begin{thm}\label{thm:owt-pe}
Let $\bar{\pi}=\left(\ln \theta+1\right)^2/\left(\ln \theta+3/4\right)<\ln \theta + 4/3$. The online
algorithm \textsf{\textsf{CR-Pursuit}}$(\bar{\pi})$ is feasible and is thus $\bar{\pi}$-competitive.
\end{thm}
\begin{proof}
With $\bar{\pi}=\left(\ln \theta+1\right)^2/\left(\ln \theta+3/4\right)$, we have $\bar{\Phi}(\bar{\pi})=\Delta$. From Lemma
\ref{lem:Phi_pi}, we know that $\Phi_\Delta(\bar{\pi})\le \bar{\Phi}(\bar{\pi})=\Delta$. Thus the results are immediate. 
\end{proof}
Note that $\bar{\pi}<\ln\theta+4/3$, which is very close to the lower bound of $\ln \theta + 1$. It also improves beyond the result of $2(\ln\theta +1)$ if we follow the characterization in Sec.~\ref{UB_CR}; the improvement is because of the tighter bound through Lemma~\ref{upper_bound_v_t_OWTPE}.

\textred{\section{Beyond the worst case mentality}
%discuss potential approach to consider beyond worst-case analysis. 
Our  \textsf{CR-Pursuit} framework focuses only on achieving competitiveness under the worst case inputs. This may limit its applications as worst case inputs or situations may seldom occur in practice. Intuitively, a ``better" online algorithm would sell more of its inventory when the incoming revenue function is ``not adversarial'', i.e., being more opportunistic. By design, \textsf{CR-Pursuit} is pessimistic: it only maintains a fixed competitive ratio $\pi^*$ for the whole trading period, even if some inputs are not adversarial. One way to improve the performance of \textsf{CR-Pursuit} for non-adversarial cases is as follows: instead of trying to keep the competitive ratio as $\pi^{*}$ during the whole period, the online algorithm \textit{adaptively} chooses a $\pi_{t}$ to maintain at time $t$. This
$\pi_{t}$ is chosen as the smallest, yet attainable, competitive ratio
at time $t$, given the previous inputs and outputs of the algorithm, and taking into account the possible inputs in future slots.
%One possible approach to improve \textsf{CR-Pursuit}, while persisting the same worst-case performance guarantee, 
%is to incooperate schostic information  Another approach, perhaps more interesting, 
%is to adaptively decrease $\pi$ in each slot based on the inputs seen so far. Intuitively, with more and more input presented, the input space is shrinking and thus the worst-case is changing, which allows the algorithm to be more aggressive. 
%Formally, at time $t$, we can define $\Phi_\Delta^t(\pi)\triangleq \max_{\sigma^{[t+1:T]}}\sum_{t=1}^Tv_t$ as the worst case inventory over all possible inputs in the remaining slots needed to maintain the competitive ratio $\pi$. By definition, we must have 
%\begin{equation}
%\label{phi_t}
%\Phi_\Delta^0(\pi)= \Phi_\Delta(\pi),
%\Phi_\Delta^t(\pi)\leq %\Phi^{t-1}_\Delta(\pi),\forall %t\in[T].
%\end{equation}
%Now, instead of choosing $\pi^*$ that satisfies $\Phi_\Delta(\pi^*)=\Delta$ and keeping $\pi^*$ as constant during the whole period (as we did in Theorem~\ref{thm:optimal_cr2}), at time $t$, we can choose $\pi_t^*$ such that $\Phi_\Delta^t(\pi_t^*)=\Delta$. From~\eqref{phi_t} and lemma~\ref{lem:vt_decreasing}, we know that $\pi_t^*\leq \pi^*$ and it is non-increasing in $t$. 
%Intuitively, with more and more input presented, the input space is shrinking and thus the worst-case is changing, which allows the algorithm to be more aggressive. Namely, 
This approach allows an online algorithm to instead pursue a competitive ratio more adaptive to the inputs, improving its average-case performance. We have recently applied the idea to develop online electric vehicle charging algorithms with optimal worst case and uniquely strong average-case performance~\cite{YLCONLINEEV-2019}. 

To better illustrate this idea, consider the following example of one-way trading. Let the first price be $p(1)=M$. The original \textsf{CR-Pursuit} algorithm sells an $\Delta/\left(\ln \theta +1\right)$ amount of inventory, and is satisfied with pursuing such competitive ratio at all time. The suggested algorithm in this section knows that the optimal offline value cannot increase any further, and would therefore sell all the inventory, i.e., we can set $\pi_1=1$. In this case, it will sell all the inventory in the first slot and achieve an offline-to-online revenue ratio of $1$ for the particular input. 
%We leave the formal analysis as future work.
%recall that $\pi$ is chosen such that \textsf{CR-Pursuit}($\pi$) can maintain the ratio under all possible inputs in the whole period. If given the inputs up to time $t$, we are sure that the considered worst-case inputs will never occurs, then we can choose another $\pi_t$ to pursuit. The new $\pi_t$ is chosen according to similar criterion:  \textsf{CR-Pursuit}($\pi_t$) can maintain the ratio under all possible inputs in remaining slots. For sure $\pi_t$ is non-increasing in $t$. 

%In terms of applications, our problem is motivated by ...

}

\section{Concluding Remarks}
\label{conclusion}
Online optimization is an important line of research with wide ranging applications.  It has been tackled by multiple algorithmic approaches over the previous decades, each proving successful for different problem variations, e.g., primal-dual approaches for online covering and packing problems or potential functions for the $k$-server problem. %Online optimization has been successfully applied to many fields, including network optimization, scheduling and so on. 

In this work, we present a \emph{novel} algorithmic framework for online optimization with inventory constraints.  The framework ``pursues'' a bound on the competitive ratio, tracking the changes in the offline optimal algorithm and ensuring that the offline-to-online revenue ratio for the instance remains bounded throughout the entire period. This idea allows us to provide an nearly optimal algorithm for online optimization with inventory constraints as well as generalizations of the classical one-way trading problem. Specifically, our analysis and algorithms generalize naturally to one-way trading problems with price elasticity and concave revenue functions, yielding almost optimal (in terms of competitive ratio) online algorithms in those settings. 

%Following the competitive ratio pursuit framework, we present the single-parametric CR-Pursuit($\pi$) algorithms and show that it is sufficient to consider this kind of algorithms for attaining the minimal competitive ratio in OOIC. When the optimal $\pi$ is calculable, we can attain the optimal algorithm. In general, we show that the optimal competitive ratio is upper bounded by $c(\ln \theta+1)$,  where $\theta $ is the ratio between the maximum and minimum base prices and $c$ is a constant that depends on the gradient properties and the maximizers of the revenue functions. In the specific case of one-way trading with price elasticity, the CR-Pursuit algorithm achieves a competitive ratio that is within $\left[\ln \theta+5/4,\ln \theta +4/3\right]$, which nearly matches the lower bound $\ln\theta +1$.

%\txtred{Paragraph on framework on limiting commitment leading to competitive ratio}

While our focus in this paper is on settings where inventory cannot be replenished,  there is a wide range of applications with both selling periods and buying periods, like battery arbitrage in contingency markets. Usually in these markets, prices are highly affected by the selling quantities and also other factors that vary in time, which lead to unknown incoming revenue functions. We believe that the  CR-Pursuit framework is promising for these problems as well, and can potentially be applicable to much broader classes of online optimization problems.  

%In general, online algorithms for two-way trading problems can be arbitrarily bad?   and stating unbounded worst cases when considering concave price elasticity. 

%The proposed pursuit algorithm maintains a constant competitive ratio, which is obtained based on worst-case inputs. However, it is too pessimistic as in practice the worst-case input seldom occurs. In these cases, we show in the Appendix that a modification can be made to the function $\Phi_\Delta$ in Theorem \ref{thm:optimal_cr2}, such that the Algorithm performs similarly on worst case inputs, but does better on other non-worst case inputs, improving average performances as compared to our original pursuit algorithm. 

For example, our focus in this paper has been on worst case analysis but the CR-Pursuit framework can also be used to provide ``beyond worst case'' results by parameterizing the bound in different ways by, for example,  utilizing properties of the instances relevant to the application, and  adaptively considering the input seen so far; see a recent example in~\cite{YLCONLINEEV-2019}. Additionally, the framework can make use of randomization when pursuing the CR bound.  This may allow improvement beyond the deterministic lower bound discussed in this paper, although it is an open question whether randomized algorithms can outperform deterministic algorithms for OOIC. 

%Other interesting extensions of this problem include considering interrelated revenue functions, or incorporating partial future information or expert predictions to improve the competitive ratio under these settings, which we leave as future work. 

%problems includes considering more practical conditions, like although the revenue functions may fluctuate along time, they are  interrelated to some extend. Another one is how to incorporate future information into algorithmic design as part of future information about revenue functions or trading period may be acquirable now. We leave these as future work.

\section{Acknowledgments}
We acknowledge the support received from Hong Kong University Grants Committee Theme-based Research Scheme Project No. T23-407/13-N and Collaborative Research Fund No. C7036-15G, NSF grant AST-134338, NSF AitF-1637598, NSF CNS-1518941, and NSF CPS-154471. In particular, John Pang wants to acknowledge the support from ASTAR, Singapore.

\bibliographystyle{ACM-Reference-Format}
\bibliography{ref}


\begin{thebibliography}{57}

%%% the Web address.

\ifx \showCODEN    \undefined \def \showCODEN     #1{\unskip}     \fi
\ifx \showDOI      \undefined \def \showDOI       #1{#1}\fi
\ifx \showISBNx    \undefined \def \showISBNx     #1{\unskip}     \fi
\ifx \showISBNxiii \undefined \def \showISBNxiii  #1{\unskip}     \fi
\ifx \showISSN     \undefined \def \showISSN      #1{\unskip}     \fi
\ifx \showLCCN     \undefined \def \showLCCN      #1{\unskip}     \fi
\ifx \shownote     \undefined \def \shownote      #1{#1}          \fi
\ifx \showarticletitle \undefined \def \showarticletitle #1{#1}   \fi
\ifx \showURL      \undefined \def \showURL       {\relax}        \fi
% invisible to TeX
\providecommand\bibfield[2]{#2}
\providecommand\bibinfo[2]{#2}
\providecommand\natexlab[1]{#1}
\providecommand\showeprint[2][]{arXiv:#2}

\bibitem[\protect\citeauthoryear{Abernethy, Hazan, and Rakhlin}{Abernethy
  et~al\mbox{.}}{2009}]%
        {abernethy2009competing}
\bibfield{author}{\bibinfo{person}{Jacob~D Abernethy}, \bibinfo{person}{Elad
  Hazan}, {and} \bibinfo{person}{Alexander Rakhlin}.}
  \bibinfo{year}{2009}\natexlab{}.
\newblock \showarticletitle{Competing in the dark: An efficient algorithm for
  bandit linear optimization}.
\newblock  (\bibinfo{year}{2009}).
\newblock


\bibitem[\protect\citeauthoryear{Akhavan-Hejazi and
  Mohsenian-Rad}{Akhavan-Hejazi and Mohsenian-Rad}{2014}]%
        {Akhavan2014Storagy}
\bibfield{author}{\bibinfo{person}{H. Akhavan-Hejazi} {and} \bibinfo{person}{H.
  Mohsenian-Rad}.} \bibinfo{year}{2014}\natexlab{}.
\newblock \showarticletitle{Optimal Operation of Independent Storage Systems in
  Energy and Reserve Markets With High Wind Penetration}.
\newblock \bibinfo{journal}{\emph{IEEE Transactions on Smart Grid}}
  \bibinfo{volume}{5}, \bibinfo{number}{2} (\bibinfo{date}{March}
  \bibinfo{year}{2014}), \bibinfo{pages}{1088--1097}.
\newblock


\bibitem[\protect\citeauthoryear{Albers}{Albers}{2003}]%
        {albers2003online}
\bibfield{author}{\bibinfo{person}{Susanne Albers}.}
  \bibinfo{year}{2003}\natexlab{}.
\newblock \showarticletitle{Online algorithms: a survey}.
\newblock \bibinfo{journal}{\emph{Mathematical Programming}}
  \bibinfo{volume}{97}, \bibinfo{number}{1-2} (\bibinfo{year}{2003}),
  \bibinfo{pages}{3--26}.
\newblock


\bibitem[\protect\citeauthoryear{Andrew, Barman, Ligett, Lin, Meyerson,
  Roytman, and Wierman}{Andrew et~al\mbox{.}}{2013}]%
        {andrew2013tale}
\bibfield{author}{\bibinfo{person}{Lachlan Andrew}, \bibinfo{person}{Siddharth
  Barman}, \bibinfo{person}{Katrina Ligett}, \bibinfo{person}{Minghong Lin},
  \bibinfo{person}{Adam Meyerson}, \bibinfo{person}{Alan Roytman}, {and}
  \bibinfo{person}{Adam Wierman}.} \bibinfo{year}{2013}\natexlab{}.
\newblock \showarticletitle{A tale of two metrics: Simultaneous bounds on
  competitiveness and regret}. In \bibinfo{booktitle}{\emph{Conference on
  Learning Theory}}. \bibinfo{pages}{741--763}.
\newblock


\bibitem[\protect\citeauthoryear{Antoniadis, Barcelo, Nugent, Pruhs, Schewior,
  and Scquizzato}{Antoniadis et~al\mbox{.}}{2016}]%
        {antoniadis2016chasing}
\bibfield{author}{\bibinfo{person}{Antonios Antoniadis}, \bibinfo{person}{Neal
  Barcelo}, \bibinfo{person}{Michael Nugent}, \bibinfo{person}{Kirk Pruhs},
  \bibinfo{person}{Kevin Schewior}, {and} \bibinfo{person}{Michele
  Scquizzato}.} \bibinfo{year}{2016}\natexlab{}.
\newblock \showarticletitle{Chasing convex bodies and functions}. In
  \bibinfo{booktitle}{\emph{Latin American Symposium on Theoretical
  Informatics}}. \bibinfo{publisher}{Springer}, \bibinfo{pages}{68--81}.
\newblock


\bibitem[\protect\citeauthoryear{Arora, Hazan, and Kale}{Arora
  et~al\mbox{.}}{2012}]%
        {arora2012multiplicative}
\bibfield{author}{\bibinfo{person}{Sanjeev Arora}, \bibinfo{person}{Elad
  Hazan}, {and} \bibinfo{person}{Satyen Kale}.}
  \bibinfo{year}{2012}\natexlab{}.
\newblock \showarticletitle{The Multiplicative Weights Update Method: a
  Meta-Algorithm and Applications.}
\newblock \bibinfo{journal}{\emph{Theory of Computing}} \bibinfo{volume}{8},
  \bibinfo{number}{1} (\bibinfo{year}{2012}), \bibinfo{pages}{121--164}.
\newblock


\bibitem[\protect\citeauthoryear{Auer, Cesa-Bianchi, Freund, and Schapire}{Auer
  et~al\mbox{.}}{1995}]%
        {auer1995gambling}
\bibfield{author}{\bibinfo{person}{Peter Auer}, \bibinfo{person}{Nicolo
  Cesa-Bianchi}, \bibinfo{person}{Yoav Freund}, {and} \bibinfo{person}{Robert~E
  Schapire}.} \bibinfo{year}{1995}\natexlab{}.
\newblock \showarticletitle{Gambling in a rigged casino: The adversarial
  multi-armed bandit problem}. In \bibinfo{booktitle}{\emph{IEEE FOCS}}.
  \bibinfo{pages}{322}.
\newblock


\bibitem[\protect\citeauthoryear{Azar, Buchbinder, Chan, Chen, Cohen, Gupta,
  Huang, Kang, Nagarajan, Naor, and Panigrahi}{Azar et~al\mbox{.}}{2016}]%
        {Azar2016convex}
\bibfield{author}{\bibinfo{person}{Y. Azar}, \bibinfo{person}{N. Buchbinder},
  \bibinfo{person}{T.~H. Chan}, \bibinfo{person}{S. Chen},
  \bibinfo{person}{I.~R. Cohen}, \bibinfo{person}{A. Gupta},
  \bibinfo{person}{Z. Huang}, \bibinfo{person}{N. Kang}, \bibinfo{person}{V.
  Nagarajan}, \bibinfo{person}{J. Naor}, {and} \bibinfo{person}{D. Panigrahi}.}
  \bibinfo{year}{2016}\natexlab{}.
\newblock \showarticletitle{Online Algorithms for Covering and Packing Problems
  with Convex Objectives}. In \bibinfo{booktitle}{\emph{2016 IEEE 57th Annual
  Symposium on Foundations of Computer Science (FOCS)}}.
  \bibinfo{pages}{148--157}.
\newblock


\bibitem[\protect\citeauthoryear{Babaioff, Immorlica, Kempe, and
  Kleinberg}{Babaioff et~al\mbox{.}}{2008}]%
        {babaioff2008online}
\bibfield{author}{\bibinfo{person}{Moshe Babaioff}, \bibinfo{person}{Nicole
  Immorlica}, \bibinfo{person}{David Kempe}, {and} \bibinfo{person}{Robert
  Kleinberg}.} \bibinfo{year}{2008}\natexlab{}.
\newblock \showarticletitle{Online auctions and generalized secretary
  problems}.
\newblock \bibinfo{journal}{\emph{ACM SIGecom Exchanges}} \bibinfo{volume}{7},
  \bibinfo{number}{2} (\bibinfo{year}{2008}).
\newblock


\bibitem[\protect\citeauthoryear{Bansa, B{\"o}hm, Eli{\'a}{\v{s}}, Koumoutsos,
  and Umboh}{Bansa et~al\mbox{.}}{2018}]%
        {bansa2018nested}
\bibfield{author}{\bibinfo{person}{Nikhil Bansa}, \bibinfo{person}{Martin
  B{\"o}hm}, \bibinfo{person}{Marek Eli{\'a}{\v{s}}},
  \bibinfo{person}{Grigorios Koumoutsos}, {and} \bibinfo{person}{Seeun~William
  Umboh}.} \bibinfo{year}{2018}\natexlab{}.
\newblock \showarticletitle{Nested convex bodies are chaseable}. In
  \bibinfo{booktitle}{\emph{Proceedings of the Twenty-Ninth Annual ACM-SIAM
  Symposium on Discrete Algorithms}}. \bibinfo{pages}{1253--1260}.
\newblock


\bibitem[\protect\citeauthoryear{Bansal, Buchbinder, and Naor}{Bansal
  et~al\mbox{.}}{2012}]%
        {bansal2012primal}
\bibfield{author}{\bibinfo{person}{Nikhil Bansal}, \bibinfo{person}{Niv
  Buchbinder}, {and} \bibinfo{person}{Joseph~Seffi Naor}.}
  \bibinfo{year}{2012}\natexlab{}.
\newblock \showarticletitle{A primal-dual randomized algorithm for weighted
  paging}.
\newblock \bibinfo{journal}{\emph{Journal of the ACM (JACM)}}
  \bibinfo{volume}{59}, \bibinfo{number}{4} (\bibinfo{year}{2012}).
\newblock


\bibitem[\protect\citeauthoryear{Bansal, Gupta, Krishnaswamy, Pruhs, Schewior,
  and Stein}{Bansal et~al\mbox{.}}{2015}]%
        {bansal20152}
\bibfield{author}{\bibinfo{person}{Nikhil Bansal}, \bibinfo{person}{Anupam
  Gupta}, \bibinfo{person}{Ravishankar Krishnaswamy}, \bibinfo{person}{Kirk
  Pruhs}, \bibinfo{person}{Kevin Schewior}, {and} \bibinfo{person}{Cliff
  Stein}.} \bibinfo{year}{2015}\natexlab{}.
\newblock \showarticletitle{A 2-competitive algorithm for online convex
  optimization with switching costs}. In
  \bibinfo{booktitle}{\emph{LIPIcs-Leibniz International Proceedings in
  Informatics}}, Vol.~\bibinfo{volume}{40}. Schloss Dagstuhl-Leibniz-Zentrum
  fuer Informatik.
\newblock


\bibitem[\protect\citeauthoryear{Bogucka, Parzy, Marques, Mwangoka, and
  Forde}{Bogucka et~al\mbox{.}}{2012}]%
        {Bogucka2012Spectrum}
\bibfield{author}{\bibinfo{person}{H. Bogucka}, \bibinfo{person}{M. Parzy},
  \bibinfo{person}{P. Marques}, \bibinfo{person}{J.~W. Mwangoka}, {and}
  \bibinfo{person}{T. Forde}.} \bibinfo{year}{2012}\natexlab{}.
\newblock \showarticletitle{Secondary spectrum trading in TV white spaces}.
\newblock \bibinfo{journal}{\emph{IEEE Communications Magazine}}
  \bibinfo{volume}{50}, \bibinfo{number}{11} (\bibinfo{date}{November}
  \bibinfo{year}{2012}), \bibinfo{pages}{121--129}.
\newblock


\bibitem[\protect\citeauthoryear{Borodin, Linial, and Saks}{Borodin
  et~al\mbox{.}}{1992}]%
        {borodin1992optimal}
\bibfield{author}{\bibinfo{person}{Allan Borodin}, \bibinfo{person}{Nathan
  Linial}, {and} \bibinfo{person}{Michael~E Saks}.}
  \bibinfo{year}{1992}\natexlab{}.
\newblock \showarticletitle{An optimal on-line algorithm for metrical task
  system}.
\newblock \bibinfo{journal}{\emph{Journal of the ACM (JACM)}}
  \bibinfo{volume}{39}, \bibinfo{number}{4} (\bibinfo{year}{1992}),
  \bibinfo{pages}{745--763}.
\newblock


\bibitem[\protect\citeauthoryear{Bubeck, Cesa-Bianchi, et~al\mbox{.}}{Bubeck
  et~al\mbox{.}}{2012}]%
        {bubeck2012regret}
\bibfield{author}{\bibinfo{person}{S{\'e}bastien Bubeck},
  \bibinfo{person}{Nicolo Cesa-Bianchi}, {et~al\mbox{.}}}
  \bibinfo{year}{2012}\natexlab{}.
\newblock \showarticletitle{Regret analysis of stochastic and nonstochastic
  multi-armed bandit problems}.
\newblock \bibinfo{journal}{\emph{Foundations and Trends{\textregistered} in
  Machine Learning}} \bibinfo{volume}{5}, \bibinfo{number}{1}
  (\bibinfo{year}{2012}), \bibinfo{pages}{1--122}.
\newblock


\bibitem[\protect\citeauthoryear{Buchbinder and Naor}{Buchbinder and
  Naor}{2005}]%
        {buchbinder2005online}
\bibfield{author}{\bibinfo{person}{Niv Buchbinder} {and}
  \bibinfo{person}{Joseph Naor}.} \bibinfo{year}{2005}\natexlab{}.
\newblock \showarticletitle{Online primal-dual algorithms for covering and
  packing problems}. In \bibinfo{booktitle}{\emph{European Symposium on
  Algorithms}}. \bibinfo{publisher}{Springer}, \bibinfo{pages}{689--701}.
\newblock


\bibitem[\protect\citeauthoryear{Buchbinder, Naor, et~al\mbox{.}}{Buchbinder
  et~al\mbox{.}}{2009}]%
        {buchbinder2009design}
\bibfield{author}{\bibinfo{person}{Niv Buchbinder},
  \bibinfo{person}{Joseph~Seffi Naor}, {et~al\mbox{.}}}
  \bibinfo{year}{2009}\natexlab{}.
\newblock \showarticletitle{The design of competitive online algorithms via a
  primal--dual approach}.
\newblock \bibinfo{journal}{\emph{Foundations and Trends{\textregistered} in
  Theoretical Computer Science}} \bibinfo{volume}{3}, \bibinfo{number}{2--3}
  (\bibinfo{year}{2009}), \bibinfo{pages}{93--263}.
\newblock


\bibitem[\protect\citeauthoryear{Chen, Agarwal, Wierman, Barman, and
  Andrew}{Chen et~al\mbox{.}}{2015}]%
        {chen2015online}
\bibfield{author}{\bibinfo{person}{Niangjun Chen}, \bibinfo{person}{Anish
  Agarwal}, \bibinfo{person}{Adam Wierman}, \bibinfo{person}{Siddharth Barman},
  {and} \bibinfo{person}{Lachlan~LH Andrew}.} \bibinfo{year}{2015}\natexlab{}.
\newblock \showarticletitle{Online convex optimization using predictions}. In
  \bibinfo{booktitle}{\emph{ACM SIGMETRICS Performance Evaluation Review}},
  Vol.~\bibinfo{volume}{43}. \bibinfo{pages}{191--204}.
\newblock


\bibitem[\protect\citeauthoryear{Chin, Fu, Guo, Han, Hu, Jiang, Lin, Ting,
  Zhang, Zhang, et~al\mbox{.}}{Chin et~al\mbox{.}}{2015}]%
        {chin2015competitive}
\bibfield{author}{\bibinfo{person}{Francis~YL Chin}, \bibinfo{person}{Bin Fu},
  \bibinfo{person}{Jiuling Guo}, \bibinfo{person}{Shuguang Han},
  \bibinfo{person}{Jueliang Hu}, \bibinfo{person}{Minghui Jiang},
  \bibinfo{person}{Guohui Lin}, \bibinfo{person}{Hing-Fung Ting},
  \bibinfo{person}{Luping Zhang}, \bibinfo{person}{Yong Zhang},
  {et~al\mbox{.}}} \bibinfo{year}{2015}\natexlab{}.
\newblock \showarticletitle{Competitive algorithms for unbounded one-way
  trading}.
\newblock \bibinfo{journal}{\emph{Theoretical Computer Science}}
  \bibinfo{volume}{607} (\bibinfo{year}{2015}), \bibinfo{pages}{35--48}.
\newblock


\bibitem[\protect\citeauthoryear{Chow, Moriguti, Robbins, and Samuels}{Chow
  et~al\mbox{.}}{1964}]%
        {chow1964optimal}
\bibfield{author}{\bibinfo{person}{YS Chow}, \bibinfo{person}{Sigaiti
  Moriguti}, \bibinfo{person}{Herbert Robbins}, {and} \bibinfo{person}{SM
  Samuels}.} \bibinfo{year}{1964}\natexlab{}.
\newblock \showarticletitle{Optimal selection based on relative rank (the
  ``secretary problem'')}.
\newblock \bibinfo{journal}{\emph{Israel Journal of mathematics}}
  \bibinfo{volume}{2}, \bibinfo{number}{2} (\bibinfo{year}{1964}),
  \bibinfo{pages}{81--90}.
\newblock


\bibitem[\protect\citeauthoryear{Damaschke, Ha, and Tsigas}{Damaschke
  et~al\mbox{.}}{2009}]%
        {damaschke2009online}
\bibfield{author}{\bibinfo{person}{Peter Damaschke},
  \bibinfo{person}{Phuong~Hoai Ha}, {and} \bibinfo{person}{Philippas Tsigas}.}
  \bibinfo{year}{2009}\natexlab{}.
\newblock \showarticletitle{Online search with time-varying price bounds}.
\newblock \bibinfo{journal}{\emph{Algorithmica}} \bibinfo{volume}{55},
  \bibinfo{number}{4} (\bibinfo{year}{2009}), \bibinfo{pages}{619--642}.
\newblock


\bibitem[\protect\citeauthoryear{Devanur and Jain}{Devanur and Jain}{2012}]%
        {Devanur2012onlinead}
\bibfield{author}{\bibinfo{person}{Nikhil~R. Devanur} {and}
  \bibinfo{person}{Kamal Jain}.} \bibinfo{year}{2012}\natexlab{}.
\newblock \showarticletitle{Online Matching with Concave Returns}. In
  \bibinfo{booktitle}{\emph{Proceedings of the Forty-fourth Annual ACM
  Symposium on Theory of Computing}} \emph{(\bibinfo{series}{STOC '12})}.
  \bibinfo{publisher}{ACM}, \bibinfo{pages}{137--144}.
\newblock


\bibitem[\protect\citeauthoryear{El-Yaniv, Fiat, Karp, and Turpin}{El-Yaniv
  et~al\mbox{.}}{2001}]%
        {el2001optimal}
\bibfield{author}{\bibinfo{person}{Ran El-Yaniv}, \bibinfo{person}{Amos Fiat},
  \bibinfo{person}{Richard~M Karp}, {and} \bibinfo{person}{Gordon Turpin}.}
  \bibinfo{year}{2001}\natexlab{}.
\newblock \showarticletitle{Optimal search and one-way trading online
  algorithms}.
\newblock \bibinfo{journal}{\emph{Algorithmica}} \bibinfo{volume}{30},
  \bibinfo{number}{1} (\bibinfo{year}{2001}), \bibinfo{pages}{101--139}.
\newblock


\bibitem[\protect\citeauthoryear{Feldman and Zenklusen}{Feldman and
  Zenklusen}{2018}]%
        {feldman2018submodular}
\bibfield{author}{\bibinfo{person}{Moran Feldman} {and} \bibinfo{person}{Rico
  Zenklusen}.} \bibinfo{year}{2018}\natexlab{}.
\newblock \showarticletitle{The submodular secretary problem goes linear}.
\newblock \bibinfo{journal}{\emph{SIAM J. Comput.}} \bibinfo{volume}{47},
  \bibinfo{number}{2} (\bibinfo{year}{2018}), \bibinfo{pages}{330--366}.
\newblock


\bibitem[\protect\citeauthoryear{Fiat}{Fiat}{1998}]%
        {fiat1998online}
\bibfield{author}{\bibinfo{person}{Amos Fiat}.}
  \bibinfo{year}{1998}\natexlab{}.
\newblock \showarticletitle{Online Algorithms: The State of the Art (Lecture
  Notes in Computer Science)}.
\newblock  (\bibinfo{year}{1998}).
\newblock


\bibitem[\protect\citeauthoryear{Fiat and Mendel}{Fiat and Mendel}{2003}]%
        {fiat2003better}
\bibfield{author}{\bibinfo{person}{Amos Fiat} {and} \bibinfo{person}{Manor
  Mendel}.} \bibinfo{year}{2003}\natexlab{}.
\newblock \showarticletitle{Better algorithms for unfair metrical task systems
  and applications}.
\newblock \bibinfo{journal}{\emph{SIAM J. Comput.}} \bibinfo{volume}{32},
  \bibinfo{number}{6} (\bibinfo{year}{2003}), \bibinfo{pages}{1403--1422}.
\newblock


\bibitem[\protect\citeauthoryear{Fiat, Rabani, and Ravid}{Fiat
  et~al\mbox{.}}{1990}]%
        {fiat1990competitive}
\bibfield{author}{\bibinfo{person}{Amos Fiat}, \bibinfo{person}{Yuval Rabani},
  {and} \bibinfo{person}{Yiftach Ravid}.} \bibinfo{year}{1990}\natexlab{}.
\newblock \showarticletitle{Competitive k-server algorithms}.
\newblock  (\bibinfo{year}{1990}).
\newblock


\bibitem[\protect\citeauthoryear{Friedman and Linial}{Friedman and
  Linial}{1993}]%
        {friedman1993convex}
\bibfield{author}{\bibinfo{person}{Joel Friedman} {and} \bibinfo{person}{Nathan
  Linial}.} \bibinfo{year}{1993}\natexlab{}.
\newblock \showarticletitle{On convex body chasing}.
\newblock \bibinfo{journal}{\emph{Discrete \& Computational Geometry}}
  \bibinfo{volume}{9}, \bibinfo{number}{3} (\bibinfo{year}{1993}),
  \bibinfo{pages}{293--321}.
\newblock


\bibitem[\protect\citeauthoryear{Fujiwara, Iwama, and Sekiguchi}{Fujiwara
  et~al\mbox{.}}{2011}]%
        {fujiwara2011average}
\bibfield{author}{\bibinfo{person}{Hiroshi Fujiwara}, \bibinfo{person}{Kazuo
  Iwama}, {and} \bibinfo{person}{Yoshiyuki Sekiguchi}.}
  \bibinfo{year}{2011}\natexlab{}.
\newblock \showarticletitle{Average-case competitive analyses for one-way
  trading}.
\newblock \bibinfo{journal}{\emph{Journal of Combinatorial Optimization}}
  \bibinfo{volume}{21}, \bibinfo{number}{1} (\bibinfo{year}{2011}),
  \bibinfo{pages}{83--107}.
\newblock


\bibitem[\protect\citeauthoryear{Grove}{Grove}{1991}]%
        {grove1991harmonic}
\bibfield{author}{\bibinfo{person}{Edward~F Grove}.}
  \bibinfo{year}{1991}\natexlab{}.
\newblock \showarticletitle{The harmonic online k-server algorithm is
  competitive}. In \bibinfo{booktitle}{\emph{Proceedings of the twenty-third
  annual ACM symposium on Theory of computing}}. \bibinfo{pages}{260--266}.
\newblock


\bibitem[\protect\citeauthoryear{Guo, Pang, and Walid}{Guo et~al\mbox{.}}{[n.
  d.]}]%
        {guo2016dynamic}
\bibfield{author}{\bibinfo{person}{Linqi Guo}, \bibinfo{person}{John Pang},
  {and} \bibinfo{person}{Anwar Walid}.} \bibinfo{year}{[n. d.]}\natexlab{}.
\newblock \showarticletitle{Dynamic service function chaining in SDN-enabled
  networks with middleboxes}. In \bibinfo{booktitle}{\emph{2016 IEEE 24th
  International Conference on Network Protocols (ICNP)}}.
  \bibinfo{pages}{1--10}.
\newblock


\bibitem[\protect\citeauthoryear{Guo, Pang, and Walid}{Guo
  et~al\mbox{.}}{2018}]%
        {guo2018joint}
\bibfield{author}{\bibinfo{person}{Linqi Guo}, \bibinfo{person}{John Pang},
  {and} \bibinfo{person}{Anwar Walid}.} \bibinfo{year}{2018}\natexlab{}.
\newblock \showarticletitle{Joint Placement and Routing of Network Function
  Chains in Data Centers}. In \bibinfo{booktitle}{\emph{IEEE INFOCOM}}.
  \bibinfo{pages}{612--620}.
\newblock


\bibitem[\protect\citeauthoryear{Hajiaghayi, Kleinberg, and
  Sandholm}{Hajiaghayi et~al\mbox{.}}{2007}]%
        {hajiaghayi2007automated}
\bibfield{author}{\bibinfo{person}{Mohammad~Taghi Hajiaghayi},
  \bibinfo{person}{Robert Kleinberg}, {and} \bibinfo{person}{Tuomas Sandholm}.}
  \bibinfo{year}{2007}\natexlab{}.
\newblock \showarticletitle{Automated online mechanism design and prophet
  inequalities}. In \bibinfo{booktitle}{\emph{AAAI}}, Vol.~\bibinfo{volume}{7}.
  \bibinfo{pages}{58--65}.
\newblock


\bibitem[\protect\citeauthoryear{Hazan, Agarwal, and Kale}{Hazan
  et~al\mbox{.}}{2007}]%
        {hazan2007logarithmic}
\bibfield{author}{\bibinfo{person}{Elad Hazan}, \bibinfo{person}{Amit Agarwal},
  {and} \bibinfo{person}{Satyen Kale}.} \bibinfo{year}{2007}\natexlab{}.
\newblock \showarticletitle{Logarithmic regret algorithms for online convex
  optimization}.
\newblock \bibinfo{journal}{\emph{Machine Learning}} \bibinfo{volume}{69},
  \bibinfo{number}{2-3} (\bibinfo{year}{2007}), \bibinfo{pages}{169--192}.
\newblock


\bibitem[\protect\citeauthoryear{Kakade, Kearns, Mansour, and Ortiz}{Kakade
  et~al\mbox{.}}{2004}]%
        {kakade2004competitive}
\bibfield{author}{\bibinfo{person}{Sham~M Kakade}, \bibinfo{person}{Michael
  Kearns}, \bibinfo{person}{Yishay Mansour}, {and} \bibinfo{person}{Luis~E
  Ortiz}.} \bibinfo{year}{2004}\natexlab{}.
\newblock \showarticletitle{Competitive algorithms for VWAP and limit order
  trading}. In \bibinfo{booktitle}{\emph{Proceedings of ACM EC}}.
  \bibinfo{pages}{189--198}.
\newblock


\bibitem[\protect\citeauthoryear{Karlin, Manasse, Rudolph, and Sleator}{Karlin
  et~al\mbox{.}}{1988}]%
        {karlin1988competitive}
\bibfield{author}{\bibinfo{person}{Anna~R Karlin}, \bibinfo{person}{Mark~S
  Manasse}, \bibinfo{person}{Larry Rudolph}, {and} \bibinfo{person}{Daniel~D
  Sleator}.} \bibinfo{year}{1988}\natexlab{}.
\newblock \showarticletitle{Competitive snoopy caching}.
\newblock \bibinfo{journal}{\emph{Algorithmica}} \bibinfo{volume}{3},
  \bibinfo{number}{1-4} (\bibinfo{year}{1988}), \bibinfo{pages}{79--119}.
\newblock


\bibitem[\protect\citeauthoryear{Li, Qiu, Ming, Quan, Qin, and Gu}{Li
  et~al\mbox{.}}{2012}]%
        {li2012online}
\bibfield{author}{\bibinfo{person}{Jiayin Li}, \bibinfo{person}{Meikang Qiu},
  \bibinfo{person}{Zhong Ming}, \bibinfo{person}{Gang Quan},
  \bibinfo{person}{Xiao Qin}, {and} \bibinfo{person}{Zonghua Gu}.}
  \bibinfo{year}{2012}\natexlab{}.
\newblock \showarticletitle{Online optimization for scheduling preemptable
  tasks on IaaS cloud systems}.
\newblock \bibinfo{journal}{\emph{J. Parallel and Distrib. Comput.}}
  \bibinfo{volume}{72}, \bibinfo{number}{5} (\bibinfo{year}{2012}),
  \bibinfo{pages}{666--677}.
\newblock


\bibitem[\protect\citeauthoryear{Li, Qu, and Li}{Li et~al\mbox{.}}{2018}]%
        {li2018using}
\bibfield{author}{\bibinfo{person}{Yingying Li}, \bibinfo{person}{Guannan Qu},
  {and} \bibinfo{person}{Na Li}.} \bibinfo{year}{2018}\natexlab{}.
\newblock \showarticletitle{Using Predictions in Online Optimization with
  Switching Costs: A Fast Algorithm and A Fundamental Limit}. In
  \bibinfo{booktitle}{\emph{IEEE Annual American Control Conference (ACC)}}.
  \bibinfo{pages}{3008--3013}.
\newblock


\bibitem[\protect\citeauthoryear{Lin, Wierman, Andrew, and Thereska}{Lin
  et~al\mbox{.}}{2013}]%
        {lin2013dynamic}
\bibfield{author}{\bibinfo{person}{Minghong Lin}, \bibinfo{person}{Adam
  Wierman}, \bibinfo{person}{Lachlan~LH Andrew}, {and} \bibinfo{person}{Eno
  Thereska}.} \bibinfo{year}{2013}\natexlab{}.
\newblock \showarticletitle{Dynamic right-sizing for power-proportional data
  centers}.
\newblock \bibinfo{journal}{\emph{IEEE/ACM Transactions on Networking (TON)}}
  \bibinfo{volume}{21}, \bibinfo{number}{5} (\bibinfo{year}{2013}),
  \bibinfo{pages}{1378--1391}.
\newblock


\bibitem[\protect\citeauthoryear{Lin, Wierman, Roytman, Meyerson, and
  Andrew}{Lin et~al\mbox{.}}{2012}]%
        {lin2012online1}
\bibfield{author}{\bibinfo{person}{Minghong Lin}, \bibinfo{person}{Adam
  Wierman}, \bibinfo{person}{Alan Roytman}, \bibinfo{person}{Adam Meyerson},
  {and} \bibinfo{person}{Lachlan~LH Andrew}.} \bibinfo{year}{2012}\natexlab{}.
\newblock \showarticletitle{Online optimization with switching cost}.
\newblock \bibinfo{journal}{\emph{ACM SIGMETRICS Performance Evaluation
  Review}} \bibinfo{volume}{40}, \bibinfo{number}{3} (\bibinfo{year}{2012}),
  \bibinfo{pages}{98--100}.
\newblock


\bibitem[\protect\citeauthoryear{L.Lu, J.Tu, C.Chau, M.Chen, and X.Lin}{L.Lu
  et~al\mbox{.}}{2013}]%
        {lu2013online}
\bibfield{author}{\bibinfo{person}{L.Lu}, \bibinfo{person}{J.Tu},
  \bibinfo{person}{C.Chau}, \bibinfo{person}{M.Chen}, {and}
  \bibinfo{person}{X.Lin}.} \bibinfo{year}{2013}\natexlab{}.
\newblock \showarticletitle{Online Energy Generation Scheduling for Microgrids
  with Intermittent Energy Sources and Co-Generation}.
\newblock \bibinfo{journal}{\emph{Proceedings of ACM Sigmetrics}}
  (\bibinfo{year}{2013}).
\newblock


\bibitem[\protect\citeauthoryear{Lorenz, Panagiotou, and Steger}{Lorenz
  et~al\mbox{.}}{2009}]%
        {lorenz2009optimal}
\bibfield{author}{\bibinfo{person}{Julian Lorenz},
  \bibinfo{person}{Konstantinos Panagiotou}, {and} \bibinfo{person}{Angelika
  Steger}.} \bibinfo{year}{2009}\natexlab{}.
\newblock \showarticletitle{Optimal algorithms for k-search with application in
  option pricing}.
\newblock \bibinfo{journal}{\emph{Algorithmica}} \bibinfo{volume}{55},
  \bibinfo{number}{2} (\bibinfo{year}{2009}), \bibinfo{pages}{311--328}.
\newblock


\bibitem[\protect\citeauthoryear{Lu, Chen, and Andrew}{Lu
  et~al\mbox{.}}{2013}]%
        {lu2013simple}
\bibfield{author}{\bibinfo{person}{Tan Lu}, \bibinfo{person}{Minghua Chen},
  {and} \bibinfo{person}{Lachlan~LH Andrew}.} \bibinfo{year}{2013}\natexlab{}.
\newblock \showarticletitle{Simple and effective dynamic provisioning for
  power-proportional data centers}.
\newblock \bibinfo{journal}{\emph{IEEE Transactions on Parallel and Distributed
  Systems}} \bibinfo{volume}{24}, \bibinfo{number}{6} (\bibinfo{year}{2013}),
  \bibinfo{pages}{1161--1171}.
\newblock


\bibitem[\protect\citeauthoryear{Mattingley, Wang, and Boyd}{Mattingley
  et~al\mbox{.}}{2011}]%
        {mattingley2011receding}
\bibfield{author}{\bibinfo{person}{Jacob Mattingley}, \bibinfo{person}{Yang
  Wang}, {and} \bibinfo{person}{Stephen Boyd}.}
  \bibinfo{year}{2011}\natexlab{}.
\newblock \showarticletitle{Receding horizon control}.
\newblock \bibinfo{journal}{\emph{IEEE Control Systems}} \bibinfo{volume}{31},
  \bibinfo{number}{3} (\bibinfo{year}{2011}), \bibinfo{pages}{52--65}.
\newblock


\bibitem[\protect\citeauthoryear{Michalska and Mayne}{Michalska and
  Mayne}{1993}]%
        {michalska1993robust}
\bibfield{author}{\bibinfo{person}{Hanna Michalska} {and}
  \bibinfo{person}{David~Q Mayne}.} \bibinfo{year}{1993}\natexlab{}.
\newblock \showarticletitle{Robust receding horizon control of constrained
  nonlinear systems}.
\newblock \bibinfo{journal}{\emph{IEEE transactions on automatic control}}
  \bibinfo{volume}{38}, \bibinfo{number}{11} (\bibinfo{year}{1993}),
  \bibinfo{pages}{1623--1633}.
\newblock


\bibitem[\protect\citeauthoryear{Pang, Fu, Lee, and Wierman}{Pang
  et~al\mbox{.}}{2017a}]%
        {pang2017efficiency}
\bibfield{author}{\bibinfo{person}{John~ZF Pang}, \bibinfo{person}{Hu Fu},
  \bibinfo{person}{Won~I Lee}, {and} \bibinfo{person}{Adam Wierman}.}
  \bibinfo{year}{2017}\natexlab{a}.
\newblock \showarticletitle{The Efficiency of Open Access in Platforms for
  Networked Cournot Markets}. In \bibinfo{booktitle}{\emph{Proceedings of IEEE
  INFOCOM}}.
\newblock


\bibitem[\protect\citeauthoryear{Pang, Guo, and Low}{Pang
  et~al\mbox{.}}{2017b}]%
        {pang2017optimal}
\bibfield{author}{\bibinfo{person}{John~ZF Pang}, \bibinfo{person}{Linqi Guo},
  {and} \bibinfo{person}{Steven~H Low}.} \bibinfo{year}{2017}\natexlab{b}.
\newblock \showarticletitle{Optimal load control for frequency regulation under
  limited control coverage}. In \bibinfo{booktitle}{\emph{IREP Symposium}}.
  \bibinfo{pages}{1--7}.
\newblock


\bibitem[\protect\citeauthoryear{Qian, Ye, Gao, Gan, Chu, Tian, Wang, and
  Guizani}{Qian et~al\mbox{.}}{2011}]%
        {Qian2011Spectrum}
\bibfield{author}{\bibinfo{person}{L. Qian}, \bibinfo{person}{F. Ye},
  \bibinfo{person}{L. Gao}, \bibinfo{person}{X. Gan}, \bibinfo{person}{T. Chu},
  \bibinfo{person}{X. Tian}, \bibinfo{person}{X. Wang}, {and}
  \bibinfo{person}{M. Guizani}.} \bibinfo{year}{2011}\natexlab{}.
\newblock \showarticletitle{Spectrum Trading in Cognitive Radio Networks: An
  Agent-Based Model under Demand Uncertainty}.
\newblock \bibinfo{journal}{\emph{IEEE Transactions on Communications}}
  \bibinfo{volume}{59}, \bibinfo{number}{11} (\bibinfo{date}{November}
  \bibinfo{year}{2011}), \bibinfo{pages}{3192--3203}.
\newblock


\bibitem[\protect\citeauthoryear{Ren, London, Ziani, and Wierman}{Ren
  et~al\mbox{.}}{2018}]%
        {ren2018datum}
\bibfield{author}{\bibinfo{person}{Xiaoqi Ren}, \bibinfo{person}{Palma London},
  \bibinfo{person}{Juba Ziani}, {and} \bibinfo{person}{Adam Wierman}.}
  \bibinfo{year}{2018}\natexlab{}.
\newblock \showarticletitle{Datum: Managing Data Purchasing and Data Placement
  in a Geo-Distributed Data Market}.
\newblock \bibinfo{journal}{\emph{IEEE/ACM Transactions on Networking}}
  \bibinfo{volume}{26}, \bibinfo{number}{2} (\bibinfo{year}{2018}),
  \bibinfo{pages}{893--905}.
\newblock


\bibitem[\protect\citeauthoryear{Rubinstein}{Rubinstein}{2016}]%
        {rubinstein2016beyond}
\bibfield{author}{\bibinfo{person}{Aviad Rubinstein}.}
  \bibinfo{year}{2016}\natexlab{}.
\newblock \showarticletitle{Beyond matroids: Secretary problem and prophet
  inequality with general constraints}. In
  \bibinfo{booktitle}{\emph{Proceedings of the forty-eighth annual ACM
  symposium on Theory of Computing}}. \bibinfo{pages}{324--332}.
\newblock


\bibitem[\protect\citeauthoryear{Schlag}{Schlag}{1998}]%
        {schlag1998imitate}
\bibfield{author}{\bibinfo{person}{Karl~H Schlag}.}
  \bibinfo{year}{1998}\natexlab{}.
\newblock \showarticletitle{Why imitate, and if so, how?: A boundedly rational
  approach to multi-armed bandits}.
\newblock \bibinfo{journal}{\emph{Journal of Economic Theory}}
  \bibinfo{volume}{78}, \bibinfo{number}{1} (\bibinfo{year}{1998}),
  \bibinfo{pages}{130--156}.
\newblock


\bibitem[\protect\citeauthoryear{Shafiee, Bhuiya, Haque, Zareipour, and
  Knight}{Shafiee et~al\mbox{.}}{2018}]%
        {Shafiee2018Battery}
\bibfield{author}{\bibinfo{person}{S. Shafiee}, \bibinfo{person}{A. Bhuiya},
  \bibinfo{person}{A.~U. Haque}, \bibinfo{person}{H. Zareipour}, {and}
  \bibinfo{person}{A.~M. Knight}.} \bibinfo{year}{2018}\natexlab{}.
\newblock \showarticletitle{Economic Assessment of Energy Storage Systems in
  Alberta's Energy and Operating Reserve Markets}. In
  \bibinfo{booktitle}{\emph{2018 IEEE/PES Transmission and Distribution
  Conference and Exposition}}. \bibinfo{pages}{1--5}.
\newblock


\bibitem[\protect\citeauthoryear{Shalev-Shwartz and Kakade}{Shalev-Shwartz and
  Kakade}{2009}]%
        {shalev2009mind}
\bibfield{author}{\bibinfo{person}{Shai Shalev-Shwartz} {and}
  \bibinfo{person}{Sham~M Kakade}.} \bibinfo{year}{2009}\natexlab{}.
\newblock \showarticletitle{Mind the duality gap: Logarithmic regret algorithms
  for online optimization}. In \bibinfo{booktitle}{\emph{Advances in Neural
  Information Processing Systems}}. \bibinfo{pages}{1457--1464}.
\newblock


\bibitem[\protect\citeauthoryear{Yang, Hajiesmaili, Yi, and Chen}{Yang
  et~al\mbox{.}}{2017}]%
        {yang2017online}
\bibfield{author}{\bibinfo{person}{Lin Yang}, \bibinfo{person}{Mohammad~H
  Hajiesmaili}, \bibinfo{person}{Hanling Yi}, {and} \bibinfo{person}{Minghua
  Chen}.} \bibinfo{year}{June 2017}\natexlab{}.
\newblock \showarticletitle{Online Offering Strategies for Storage-Assisted
  Renewable Power Producer in Hour-Ahead Market}. In
  \bibinfo{booktitle}{\emph{Proceedings of ACM SIGMETRICS}}.
\newblock


\bibitem[\protect\citeauthoryear{Yi*, Lin*, and Chen}{Yi*
  et~al\mbox{.}}{2019}]%
        {YLCONLINEEV-2019}
\bibfield{author}{\bibinfo{person}{Hanling Yi*}, \bibinfo{person}{Qiulin Lin*},
  {and} \bibinfo{person}{Minghua Chen}.} \bibinfo{year}{2019}\natexlab{}.
\newblock \showarticletitle{Balancing Cost and Dissatisfaction in Online EV
  Charging under Real-time Pricing}. In \bibinfo{booktitle}{\emph{IEEE
  INFOCOM}}.
\newblock
\newblock
\shownote{(The first two authors contribute equally to the work.).}


\bibitem[\protect\citeauthoryear{Zhang, Xu, Zheng, and Dong}{Zhang
  et~al\mbox{.}}{2012}]%
        {zhang2012optimal}
\bibfield{author}{\bibinfo{person}{Wenming Zhang}, \bibinfo{person}{Yinfeng
  Xu}, \bibinfo{person}{Feifeng Zheng}, {and} \bibinfo{person}{Yucheng Dong}.}
  \bibinfo{year}{2012}\natexlab{}.
\newblock \showarticletitle{Optimal algorithms for online time series search
  and one-way trading with interrelated prices}.
\newblock \bibinfo{journal}{\emph{Journal of combinatorial optimization}}
  \bibinfo{volume}{23}, \bibinfo{number}{2} (\bibinfo{year}{2012}),
  \bibinfo{pages}{159--166}.
\newblock


\bibitem[\protect\citeauthoryear{Zhang, Hajiesmaili, Cai, Chen, and Zhu}{Zhang
  et~al\mbox{.}}{2018}]%
        {zhang2018peak}
\bibfield{author}{\bibinfo{person}{Ying Zhang}, \bibinfo{person}{Mohammad~H
  Hajiesmaili}, \bibinfo{person}{Sinan Cai}, \bibinfo{person}{Minghua Chen},
  {and} \bibinfo{person}{Qi Zhu}.} \bibinfo{year}{2018}\natexlab{}.
\newblock \showarticletitle{Peak-aware online economic dispatching for
  microgrids}.
\newblock \bibinfo{journal}{\emph{IEEE Transactions on Smart Grid}}
  \bibinfo{volume}{9}, \bibinfo{number}{1} (\bibinfo{year}{2018}),
  \bibinfo{pages}{323--335}.
\newblock


\end{thebibliography}

% \newpage{}

\appendix

\section{Appendix} 

\subsection{A Binary Search Algorithm for Computing $\lambda^{*}$} \label{apx:offline.optimal.alg}
We summarize the algorithm in Algorithm~\ref{alg:Offline-algorithm}.
\begin{algorithm}[!ht]
\caption{A Binary search algorithm for Computing $\lambda^{*}$ \label{alg:Offline-algorithm}}
\begin{algorithmic}[1]

\IF{$\max_{v_t\in V_t(0)}\sum_{t=1}^Tv_t\leq\Delta$}

\RETURN $\lambda^{*}=0$;

\ELSE 

\STATE Pick $\lambda_{L}=0,$ $\lambda_{H}=\max_{t\in\mathcal{T}}\left(g'_t(0)\right)$;

\WHILE{$\left|\lambda_{L}-\lambda_{H}\right|>\epsilon$} 

\STATE $\lambda_{M}=\frac{\lambda_{L}+\lambda_{H}}{2},v_t=0,\forall t\in\mathcal{T}$;

\STATE %Compute $v_t$ according to (\ref{eq:vt-1});
Compute $\Delta_{\max}=\max_{v_t\in V_t(\lambda_M)}\sum_{t=1}^Tv_t$ \STATE Compute $\Delta_{\min}=\min_{v_t\in V_t(\lambda_M)}\sum_{t=1}^Tv_t$.

\IF{$\Delta_{\min}>\Delta$} 

\STATE $\lambda_{L}=\lambda_{M}$;

\ENDIF 

\IF{$\Delta_{\max}<\Delta$} 

\STATE $\lambda_{H}=\lambda_{M}$;

\ENDIF

\IF{$\Delta_{\min}\le\Delta\le\Delta_{\max}$}
\STATE break;
\ENDIF

\ENDWHILE 

\RETURN $\lambda^{*}=\lambda_{M}$;

\ENDIF 

\end{algorithmic}
\end{algorithm}

\subsection{Proof of Proposition \ref{thm:offline_solution}}
\begin{proof}
We prove this theorem by investigating the KKT conditions of problem
OOIC and exploring the structure of the optimal solution.

The Lagrangian for problem OOIC is defined as 
\[
L\left(v,\lambda,\mu\right)=\sum_{t=1}^{T}g_{t}(v_t)+\lambda\left(\Delta-\sum_{t=1}^{T}v_t\right)+\sum_{t=1}^{T}v_t\mu(t),
\]
where $\lambda\ge0$ and $\mu(t)\ge0$, $\forall t\in[T]$
are the Lagrangian multipliers. The following KKT conditions 
give us a set of necessary and sufficient conditions for optimality: 
\begin{align*}
    g'_{t}(v_t)-\lambda+\mu(t)= & 0, \quad\forall t\in[T],    \\
    \sum_{t=1}^{T}v_t\leq & \Delta,     \\
    v_t\geq & 0,\quad\forall t\in[T],    \\
    \mu(t)\geq& 0,\quad\forall t\in[T],  \\
    \lambda\geq & 0,    \\
    v_t\mu(t)=& 0,\quad\forall t\in[T],  \\
    \lambda\left(\sum_{t=1}^{T}v_t-\Delta\right)= & 0.
\end{align*}
Suppose $v^{*}$, $\mu^{*}$ and $\lambda^{*}$ are the optimal solutions
that satisfy the KKT conditions. Denote the set $\mathcal{T}_{0}=\{t|v^{*}_t>0,\forall t\in[T]\}$,
then according to the KKT conditions, we have 
\begin{align}
\mu^{*}(t) & =0,\quad\forall t\in\mathcal{T}_{0},\label{eq:ut}\\
\lambda^*\left(\sum_{t\in\mathcal{T}_{0}}v^{*}_t-\Delta\right) & =0,\label{eq:sum}\\
g'_{t}(v^{*}_t)-\lambda^{*} & =0,\quad\forall t\in\mathcal{T}_{0},\label{eq:vt}
\end{align}

Since $g'_t$ is concave, $g'_{t}(\cdot)$ is non-increasing in $v_t$. According to (\ref{eq:vt}) we have 
\[
g'_t(0)\geq g'_t(v^{*}_t)=\lambda^{*},\quad\forall t\in\mathcal{T}_{0};
\]
namely, 
\begin{equation}
g'_t(0)\ge\lambda^{*}\quad\forall t\in\mathcal{T}_{0}.\label{eq:T0}
\end{equation}
Thus given a $\lambda^{*}$, we can use (\ref{eq:T0}) to determine
the set $\mathcal{T}_{0}$.

For ease of presentation, we denote $$V_t(\lambda)=\{v|g'_t(v)=\lambda,v\in[0,\Delta]\}.$$ Now consider the following two cases:

(1) $\Delta\ge\max_{v_t\in V_t(0)}\sum_{t=1}^T v_t$.
In this case, we observe that the solution 
\begin{align*}
v^{*}_t & \in V_t(0),\forall t\in[T],\\
\lambda^{*} & =0,\\
\mu^{*}(t) & =0,\forall t\in[T],
\end{align*}
satisfies the KKT conditions, thus it is the optimal solution.

(2) $\Delta<\max_{v_t\in V_t(0)}\sum_{t=1}^T v_t$. In this case, we must have $\lambda^*>0$.
According to (\ref{eq:sum}) and (\ref{eq:vt}), we have 
\begin{equation*}
    v^*_t \in  V_t(\lambda^*) \mbox{ and } \sum_{t=1}^Tv^*_t=\Delta.
\end{equation*}

It is straightforward to check that  $v_t,\forall t\in\mathcal{T}_{0}$,
is non-increasing w.r.t. $\lambda$. Meanwhile, according
to (\ref{eq:T0}), we know that the size of set $\mathcal{T}_{0}$
is non-increasing w.r.t. $\lambda$. Putting together these two observations,
we conclude that $\sum_{t\in\mathcal{T}_{0}}v_t$ is non-increasing w.r.t. $\lambda$. Thus given $\Delta>0$, there exists
a unique $\lambda=\lambda^{*}$ that satisfies $\sum_{t\in\mathcal{T}_{0}}v^{*}_t=\Delta$.
Since KKT conditions are necessary and sufficient for optimality of
convex problems, we can conclude that $\lambda^{*}$ is the optimal
dual solution.
\end{proof}

\iffalse
\subsection{Proof of Corollary \ref{prop:sufficient}}
\begin{proof}[Proof of Corollary \ref{prop:sufficient}]
Let $\bar{v}_{\tau},\tau\in[t-1]$ be the optimal solution under input $\sigma^{[1:t-1]}$,
and $v_{\tau}^*,\tau\in[t]$ be the optimal solution at under input $\sigma^{[1:t]}$. 
Then we have 
\begin{eqnarray*}
 & \eta_{OPT}(\sigma^{[1:t]})-\eta_{OPT}(\sigma^{[1:t-1]}) & =g_{t}(v_{t}^*)+\sum_{\tau=1}^{t-1}(g_\tau(v_{\tau}^*)-g_\tau(\bar{v}_{\tau}))\\
 &  & \leq g_{t}(v_{t}^*)\leq g_{t}(\hat{v}_{t})
\end{eqnarray*}
\end{proof}
\fi

\subsection{Proof of Lemma \ref{lem:opt_difference_bound}}
\label{app:proof_opt_difference_bound}
\begin{proof}[Proof of Lemma \ref{lem:opt_difference_bound}]
%When $v({\lambda_t})=0$ or $v({\lambda_{t-1}})=0$, we can see that $\eta_{OPT}(\sigma^{[1:t]})-\eta_{OPT}(\sigma^{[1:t-1]})=0$ and hence it holds directly. If $v({\lambda_{t-1}})=\Delta$, meaning $g'_t(v)> \lambda_{t-1}$. We then focus on the case that $v({\lambda_t})>0$ or $v({\lambda_{t-1}})>0$. In this case, we have $\lambda_{t}=g_t'(v({\lambda_{t}}))$ and $\lambda_{t-1}=g_t'(v({\lambda_{t-1}}))$. 
We prove this lemma in the following two steps:

Step I, we prove that  $\eta_{OPT}\left(\sigma^{[1:t]}\right)-\eta_{OPT}\left(\sigma^{[1:t-1]}\right)\geq g_t(\tilde{v}_t)-\lambda_{t}\tilde{v}_t$. To see this, we denote optimal solution at time $\tau\in[t]$ under input $\sigma^{[1:t]}$ as $\tilde{v}_{\tau}.$ Note that $\tilde{v}_{\tau}\in V_\tau(\lambda_t),\tau\in[t]$ or  $\tilde{v}_{\tau}=0$  if $V_\tau(\lambda_t)=\emptyset $.  
Similarly, denote optimal solution at time $\tau\in[t-1]$ under input $\sigma^{[1:t-1]}$ as $\bar{v}_{\tau}.$ Note that $\bar{v}_{\tau}\in V_\tau(\lambda_{t-1}),\tau\in[t-1]$ or $\bar{v}_{\tau}=0$ if $V_\tau(\lambda_{t-1})=\emptyset$. Also $\tilde{v}_{\tau}\leq\bar{v}_{\tau},\tau\in[t-1]$ (by the non-increasing of $g_{t}'(v)$ and $\lambda_{t}\geq\lambda_{t-1}$). Then we have
\begin{align*}
 %&  & \eta_{OPT}\left(\sigma^{[1:t]}\right)-\eta_{OPT}\left(\sigma^{[1:t-1]}\right)\\
 \eta_{OPT}\left(\sigma^{[1:t]}\right)-\eta_{OPT}\left(\sigma^{[1:t-1]}\right) = & \sum_{\tau=1}^{t}g_{\tau}(\tilde{v}_{\tau})-\sum_{\tau=1}^{t-1}g_{\tau}(\bar{v}_{\tau})\\
  = & g_{t}(\tilde{v}_{t})+\sum_{\tau=1}^{t-1}\left(g_{\tau}(\tilde{v}_{\tau})-g_{\tau}(\bar{v}_{\tau})\right)\\
  \stackrel{(a)}{\geq} & g_{t}(\tilde{v}_{t})+\sum_{\tau=1}^{t-1}\lambda_{t}(\tilde{v}_{\tau}-\bar{v}_{\tau})\\
  \stackrel{(b)}{\geq} & g_{t}(\tilde{v}_{t})-\lambda_{t}\tilde{v}_{t}. %\\
%& = & g_{t}({v}({\lambda_t}))-\lambda_{t}{v}({\lambda_t})
\end{align*}
For (a), it comes from the concavity of $g_{\tau}(v)$ and $\tilde{v}_{\tau}\leq\bar{v}_{\tau},\tau\in[t-1]$.
For (b), we claim that $\sum_{\tau=1}^{t-1}\bar{v}_{\tau}\leq\sum_{\tau=1}^{t}\tilde{v}_{\tau}$. To see this, when $\lambda_t=0$, we must have $\lambda_{t-1}=0$. In this case, $\tilde{v}_\tau=\bar{v}_\tau,\forall \tau\in [t-1]$ and thus we have $\sum_{\tau=1}^{t-1}\bar{v}_{\tau}\leq\sum_{\tau=1}^{t}\tilde{v}_{\tau}.$  When $\lambda_t>0$, from the KKT conditions in \eqref{eq:sum}, we have $\sum_{\tau=1}^{t}\tilde{v}_{\tau}=\Delta\geq \sum_{\tau=1}^{t-1}\bar{v}_{\tau}$. 
Then we conclude that $\sum_{\tau=1}^{t-1}\bar{v}_{\tau}\leq\sum_{\tau=1}^{t}\tilde{v}_{\tau}$ and consequently, we have  $\sum_{\tau=1}^{t-1}(\tilde{v}_{\tau}-\bar{v}_{\tau})\geq -\tilde{v}_{t}$.

Step II, we prove that $\eta_{OPT}\left(\sigma^{[1:t]}\right)-\eta_{OPT}\left(\sigma^{[1:t-1]}\right)\leq g_t(\tilde{v}_t)-\lambda_{t-1}\tilde{v}_t\leq g_{t}(\tilde{v}_{t})\leq g_{t}(\hat{v}_{t})$. Similarly, we have 
\begin{align*}
% &  & \eta_{OPT}\left(\sigma^{[1:t]}\right)-\eta_{OPT}\left(\sigma^{[1:t-1]}\right)\\
 \eta_{OPT}\left(\sigma^{[1:t]}\right)-\eta_{OPT}\left(\sigma^{[1:t-1]}\right) = & g_{t}(\tilde{v}_{t})+\sum_{\tau=1}^{t-1}(g_{\tau}(\tilde{v}_{\tau})-g_{\tau}(\bar{v}_{\tau}))\\
  \stackrel{(a)}{\leq} & g_{t}(\tilde{v}_{t})+\sum_{\tau=1}^{t-1}\lambda_{t-1}(\tilde{v}_{\tau}-\bar{v}_{\tau})\\
  \stackrel{(b)}{=} & g_{t}(\tilde{v}_{t})-\lambda_{t-1}\tilde{v}_{t}\\
   \leq & g_{t}(\tilde{v}_{t})\leq g_{t}(\hat{v}_{t}).
\end{align*}
For (a), it is by the concavity of ${g_\tau}$: $g_{\tau}(\tilde{v}_{\tau})\leq g_{\tau}(\bar{v}_{\tau})+\lambda_{\tau-1}(\tilde{v}_{\tau}-\bar{v}_{\tau})$  (Note that $\lambda_{\tau}= g_\tau'(\bar{v}_{\tau})$) and $\lambda_{t-1}\geq \lambda_\tau,\forall \tau\in[t-1]$.
For (b), when $\lambda_{t-1}=0$, it holds immediately; when $\lambda_{t-1} > 0$, we have $\sum_{\tau}^{t}\tilde{v}_\tau=\Delta=\sum_{\tau=1}^{t-1}\bar{v}_\tau$, which implies $\sum_{\tau=1}^{t-1}(\tilde{v}_{\tau}-\bar{v}_{\tau}) =- \tilde{v}_{t}$.
%For (c), it comes from the concavity of $g_t$: $g_{t}({v}_{t})\leq g_{t}({v}({\lambda_{t-1}}))+\lambda_{t-1}({v}_{t}-{v}({\lambda_{t-1}}))$  (Note that $\lambda_{t-1}=g_t'(v({\lambda_{t-1}}))$). 
\end{proof}

\subsection{Proof of Lemma~\ref{lem:interchange_input_opt}}
\label{app:proof_interchange_input_opt}
\begin{proof}
Denote the input under $\tilde{\sigma}$ as ${g}_t$.
Denote the input under $\bar{\sigma}$ as $\bar{g}_t$,  The optimal dual variable under $\tilde{\sigma}^{[1:t]}$ (resp. $\bar{\sigma}^{[1:t]}$ ) as $\lambda_t$ (resp. $\bar{\lambda}_t$). We have,
\[
g_t=\bar{g}_t, \forall t\leq \tau-1 \lor t \geq \tau+2.
\]
Besides, $g_\tau=\bar{g}_{\tau+1},g_{\tau+1}=\bar{g}_\tau$. Let $v_t$ (resp. $\bar{v}_t$) be the optimal offline solution at time $t$ given the input $\tilde{\sigma}^{[1:t]}$ (resp. $\bar{\sigma}^{[1:t]}$).

1) If $\lambda_\tau \leq \bar{\lambda}_\tau$, then
\begin{align*}
% & \eta_{OPT}\left(\tilde{\sigma}^{[1:\tau]}\right)-\eta_{OPT}\left(\tilde{\sigma}^{[1:\tau-1]}\right)\\
\eta_{OPT}\left(\tilde{\sigma}^{[1:\tau]}\right)-\eta_{OPT}\left(\tilde{\sigma}^{[1:\tau-1]}\right)\stackrel{(a)}{\geq} & g_\tau(v_\tau)-\lambda_{\tau}v_\tau  \\
\stackrel{(b)}{\geq} & g_\tau(\bar{v}_{\tau+1})-{\lambda}_{\tau}\bar{v}_{\tau+1}\\
 \stackrel{(c)}{\geq}  & g_\tau(\bar{v}_{\tau+1})-\bar{\lambda}_{\tau}\bar{v}_{\tau+1} \\
\stackrel{(a)}{\geq} & \eta_{OPT}\left(\bar{\sigma}^{[1:\tau+1]}\right)-\eta_{OPT}\left(\bar{\sigma}^{[1:\tau]}\right).
\end{align*}
For (a), it is by lemma~\ref{lem:opt_difference_bound}. For (b), it is by the concavity of $g_t$ and for (c), it by $\lambda_\tau \leq \bar{\lambda}_\tau$.

2) If $\lambda_\tau \geq \bar{\lambda}_\tau$, then similarly
\begin{align*}
\eta_{OPT}\left(\bar{\sigma}^{[1:\tau]}\right)-\eta_{OPT}\left(\bar{\sigma}^{[1:\tau-1]}\right)
\stackrel{(a)}{\geq} & g_{\tau+1}(\bar{v}_{\tau})-\bar{\lambda}_{\tau}\bar{v}_{\tau}  \\
 \stackrel{(b)}{\geq}  & g_{\tau+1}(\bar{v}_{\tau})-{\lambda}_{\tau}\bar{v}_{\tau} \\
 \stackrel{(c)}{\geq}  & g_{\tau+1}(v_{\tau+1})-{\lambda}_{\tau}v_{\tau+1} \\
\stackrel{(a)}{\geq} & \eta_{OPT}\left(\tilde{\sigma}^{[1:\tau+1]}\right)-\eta_{OPT}\left(\tilde{\sigma}^{[1:\tau]}\right).
\end{align*}
For (a), it is by lemma~\ref{lem:opt_difference_bound}. For (b), it is by $\lambda_\tau \geq \bar{\lambda}_\tau$. For (c), it is by the concavity of $g_t$. Also, with
\begin{align*}
    &\eta_{OPT}\left(\bar{\sigma}^{[1:\tau]}\right)-\eta_{OPT}\left(\bar{\sigma}^{[1:\tau-1]}\right)+\eta_{OPT}\left(\bar{\sigma}^{[1:\tau+1]}\right)-\eta_{OPT}\left(\bar{\sigma}^{[1:\tau]}\right)\\
    = & \eta_{OPT}\left(\bar{\sigma}^{[1:\tau+1]}\right)-\eta_{OPT}\left(\bar{\sigma}^{[1:\tau-1]}\right)\\
    = &\eta_{OPT}\left(\tilde{\sigma}^{[1:\tau+1]}\right)-\eta_{OPT}\left(\tilde{\sigma}^{[1:\tau-1]}\right)\\
    = & \eta_{OPT}\left(\tilde{\sigma}^{[1:\tau+1]}\right)-\eta_{OPT}\left(\tilde{\sigma}^{[1:\tau]}\right)+\eta_{OPT}\left(\tilde{\sigma}^{[1:\tau]}\right)-\eta_{OPT}\left(\tilde{\sigma}^{[1:\tau-1]}\right),
\end{align*}
we can have
\[
\eta_{OPT}\left(\bar{\sigma}^{[1:\tau+1]}\right)-\eta_{OPT}\left(\bar{\sigma}^{[1:\tau]}\right)\leq \eta_{OPT}\left(\tilde{\sigma}^{[1:\tau]}\right)-\eta_{OPT}\left(\tilde{\sigma}^{[1:\tau-1]}\right). 
\]
\end{proof}

\subsection{Proof of Lemma~\ref{lem:increasing_marginal_worstcase}}
\label{app:proof_increasing_marginal_worstcase}
\begin{proof}[Proof of Lemma \ref{lem:increasing_marginal_worstcase}]
%We prove this lemma by contradiction. 
Suppose an arbitrary $\tilde{\sigma}\in\arg\max_\sigma \sum_t{v_t}$, under which $g_t'(v_t)$ is not non-decreasing in $t$, where $v_t$ is the selling quantity of CR-Pursuit($\pi$) under $\tilde{\sigma}$. That is, exist a $\tau$, $g_{\tau}'(v_\tau) > g_{\tau+1}'(v_{\tau+1})$. Denote the optimal dual variables under $\tilde{\sigma}^{[1:t]}$ as $\lambda_t$. Note that $\lambda_t$ is non-decreasing in $t$. Without loss of generality, we assume that $\lambda_t < \lambda_{t+1}$ or $\lambda_t=\lambda_{t+1}=0$, $\forall t$. We construct a new input sequence $\bar{\sigma}$ by interchanging $g_\tau$ and $g_{\tau+1}$ in $\tilde{\sigma}$ and denote the input under $\bar{\sigma}$ as $\bar{g}_t$, the output of CR-Pursuit$(\pi^{*})$ under $\bar{\sigma}$ as $\bar{v}_t$. The optimal dual variable under $\bar{\sigma}^{[1:t]}$ as $\bar{\lambda}_t$. By definition, we can easily observe that,
\[
\eta_{OPT}\left(\tilde{\sigma}^{[1:t]}\right)=\eta_{OPT}\left(\bar{\sigma}^{[1:t]}\right), \forall t\leq \tau-1 \lor t \geq \tau+1;
\]
\[
v_t=\bar{v}_t, \forall t\leq \tau-1 \lor t \geq \tau+2;
\]
\[
g_t=\bar{g}_t, \forall t\leq \tau-1 \lor t \geq \tau+2.
\]
Besides, $g_\tau=\bar{g}_{\tau+1},g_{\tau+1}=\bar{g}_\tau$. We claim that $\bar{\sigma}\in\arg\max_\sigma \sum_t{v_t}$ and $\bar{g}_{\tau}'(\bar{v}_\tau)=g_{\tau+1}'(v_{\tau+1}) < g_{\tau}'({v}_{\tau})=\bar{g}_{\tau+1}'(\bar{v}_{\tau+1})$. 
%(equivalently $g_{\tau+1}'(\bar{v}_\tau) \leq g_{\tau}'(\bar{v}_{\tau+1})$).  
To see this, consider the following two cases.

(1) $\lambda_\tau=\lambda_{\tau+1}=0$. Under this case, we have \begin{align*}
\eta_{OPT}\left(\tilde{\sigma}^{[1:\tau]}\right)-\eta_{OPT}\left(\tilde{\sigma}^{[1:\tau-1]}\right)&=\eta_{OPT}\left(\bar{\sigma}^{[1:\tau+1]}\right)-\eta_{OPT}\left(\bar{\sigma}^{[1:\tau]}\right)\\
&=g_\tau(\hat{v}_\tau),
\end{align*}
where $\hat{v}_\tau=\arg\max_v g_\tau(v)$. Then $v_\tau=\bar{v}_{\tau+1}$. Similarly, we have $v_{\tau+1}=\bar{v}_\tau$. $\sum_t v_t=\sum_t \bar{v}_t$. We conclude that $\bar{\sigma}\in\arg\max_\sigma \sum_t{v_t}$ and $\bar{g}_{\tau}'(\bar{v}_\tau)=g_{\tau+1}'(v_{\tau+1}) < g_{\tau}'({v}_{\tau})=\bar{g}_{\tau+1}'(\bar{v}_{\tau+1})$.

(2) $0\leq \lambda_\tau < \lambda_{\tau+1}$. First, we have
\begin{align*}
 g_\tau(v_\tau)+g_{\tau+1}(v_{\tau+1})&=\frac{\eta_{OPT}\left(\tilde{\sigma}^{[1:\tau+1]}\right)-\eta_{OPT}\left(\tilde{\sigma}^{[1:\tau-1]}\right)}{\pi^*}\\  
   & =\frac{\eta_{OPT}\left(\bar{\sigma}^{[1:\tau+1]}\right)-\eta_{OPT}\left(\bar{\sigma}^{[1:\tau-1]}\right)}{\pi^*}\\
   & =g_{\tau+1}(\bar{v}_\tau)+g_{\tau}(\bar{v}_{\tau+1}),
\end{align*}
which implies 
\[
 g_\tau(v_\tau)-g_{\tau}(\bar{v}_{\tau+1})=g_{\tau+1}(\bar{v}_{\tau})-g_{\tau+1}(v_{\tau+1}).
\]
Second, we claim that $\bar{v}_{\tau+1} \leq v_{\tau}$. From Lemma~\ref{lem:interchange_input_opt}, we have $$\eta_{OPT}\left(\bar{\sigma}^{[1:\tau+1]}\right)-\eta_{OPT}\left(\bar{\sigma}^{[1:\tau]}\right)\leq \eta_{OPT}\left(\tilde{\sigma}^{[1:\tau]}\right)-\eta_{OPT}\left(\tilde{\sigma}^{[1:\tau-1]}\right).$$ Then
$g_\tau(v_\tau)\geq g_\tau(\bar{v}_{\tau+1})$ and $\bar{v}_{\tau+1}\leq v_\tau$ are straightforward.

Third, we show $g_\tau(v_\tau) = g_\tau(\bar{v}_{\tau+1})$ and thus $\bar{v}_{\tau+1} = v_\tau$ by contradiction. Suppose $g_\tau(v_\tau) > g_\tau(\bar{v}_{\tau+1})$ and thus $\bar{v}_{\tau+1}< v_\tau$. we show that $\sum_t{v_t}< \sum_t{\bar{v}_t}$ which contradict the fact that $\tilde{\sigma}\in \arg\max_\sigma \sum_t{v_t}$. To see this, observe that we have
\begin{align*}
   g_{\tau+1}'(v_{\tau+1})(v_\tau-\bar{v}_{\tau+1}) \stackrel{(a)}{<} &-g_\tau'(v_\tau)(\bar{v}_{\tau+1}-v_\tau)\\
    \stackrel{(b)}{\leq} & g_\tau(v_\tau)-g_{\tau}(\bar{v}_{\tau+1})\\
    =    &g_{\tau+1}(\bar{v}_{\tau})-g_{\tau+1}(v_{\tau+1})\\
    \stackrel{(b)}{\leq} & g_{\tau+1}'(v_{\tau+1})(\bar{v}_{\tau}-v_{\tau+1}).
\end{align*}
For (a), it is by $g_\tau'(v_\tau) > g_{\tau+1}'(v_{\tau+1}) \geq \lambda_{t+1}> 0$ and $\bar{v}_{\tau+1} < v_{\tau}$. For (b), it is from the concavity of ${g_\tau}$. 
%By It implies $$g_{\tau+1}'(v_{\tau+1})(v_\tau-\bar{v}_\tau)\leq g_{\tau+1}'(v_{\tau+1})(\bar{v}_{\tau}-v_{\tau+1}),$$
As $g_{\tau+1}'(v_{\tau+1})\geq \lambda_{\tau+1}>0$, we have
\[
v_\tau+v_{\tau+1} < \bar{v}_\tau+\bar{v}_{\tau+1},
\]
which leads to $\sum_t v_t < \sum_t \bar{v}_t$. 

So we conclude that $g_\tau(v_\tau) = g_\tau(\bar{v}_{\tau+1})$ and thus $\bar{v}_{\tau+1} = v_\tau$. Consequently, $g_{\tau}(v_{\tau+1}) = g_{\tau+1}(\bar{v}_{\tau})$ and thus $\bar{v}_{\tau} = v_{\tau+1}$. It is then straightforward that 
\[
\bar{\sigma}\in\arg\max_\sigma \sum_t{v_t},
\] and
\[
\bar{g}_{\tau}'(\bar{v}_\tau)=g_{\tau+1}'(v_{\tau+1}) < g_{\tau}'({v}_{\tau})=\bar{g}_{\tau+1}'(\bar{v}_{\tau+1}).
\]

\iffalse
As for $\bar{g}_{\tau}'(\bar{v}_\tau) \leq \bar{g}_{\tau+1}'(\bar{v}_{\tau+1})$ (equivalently $g_{\tau+1}'(\bar{v}_\tau) \leq g_{\tau}'(\bar{v}_{\tau+1})$). It can be showed as the following,
\begin{align*}
& g'_{t}(\bar{v}_{\tau+1})(\bar{v}_\tau-v_{\tau+1})\\
\stackrel{(a)}{\geq}   & g'_{\tau}(\bar{v}_{\tau+1})(v_\tau-\bar{v}_{\tau+1})\\
    \stackrel{(b)}{\geq} & g_{\tau}(v_\tau)-g_{\tau}(\bar{v}_{\tau+1})\\
    = & g_{\tau+1}(\bar{v}_{\tau})-g_{\tau+1}(v_{\tau+1})\\
    \stackrel{(b)}{\geq} & g'_{\tau+1}(\bar{v}_{\tau})(\bar{v}_{\tau}-v_{\tau+1}) 
\end{align*}
(a) is by $0\leq \bar{v}_\tau-v_{\tau+1}\leq\bar{v}_\tau-v_{\tau+1}$. (b) is by the concavity of ${g_t}$. We then easily conclude $g'_{\tau+1}(\bar{v}_\tau) \leq g'_{\tau}(\bar{v}_{\tau+1})$. 
\fi

By continuously interchanging $g_\tau$  and $g_{\tau+1}$ that fails to satisfy $g'_{\tau+1}({v}_\tau) \leq g'_{\tau}(v_{\tau+1})$, we finally attain a sequence in $\arg\max_\sigma \sum_t{v_t}$ such that $g_t'(v_t)$ is non-decreasing in $t$.
\end{proof}

\subsection{Proof of Lemma~\ref{lem:upper_bound_v_t_OOIC}}\label{apx:proof.lem.upper_bound_v_t_OOIC}
\begin{proof}
First, from Lemma \ref{lem:opt_difference_bound}, we easily conclude that $\bar{v}_{t}\leq \hat{v}_{t}$, where $\hat{v}_t$ is the optimizer of $g_t(\cdot)$. 
% For easy of representation, we denote $k=g_{t}(\bar{v}_{t})=\frac{\Delta(\tilde{p}(t)-\tilde{p}(t-1))}{\pi}$. Then 
By the concavity of $g_t(\cdot)$, we have
\[
g_t(\bar{v}_t)\geq \frac{\bar{v}_t}{\hat{v}_t}g_t\left(\hat{v}_t\right)+\left(1-\frac{\bar{v}_t}{\hat{v}_t}\right)g_t(0)\geq \frac{\bar{v}_t}{\hat{v}_t}g_t\left(\hat{v}_t\right),
\]
which then gives
$
\bar{v}_t\leq \frac{g_t(\bar{v}_t)}{g_t\left(\hat{v}_t\right)/\hat{v}_t}.
$
Then, using the definition of $c$, we arrive at
\[
\bar{v}_t\leq \frac{p(t)}{g_t\left(\hat{v}_t\right)/\hat{v}_t}\frac{g_t(\bar{v}_t)}{p(t)}\leq c\frac{g_t(\bar{v}_t)}{p(t)}.
\]
\end{proof}
% Combining the above two inequality, we conclude 
% \[
% \bar{v}_t\leq c\frac{\Delta}{\pi}\frac{(\tilde{p}(t)-\tilde{p}(t-1))}{p(t)}
% \]
% \iffalse
% So if $c\frac{\Delta}{\pi}\frac{(\tilde{p}(t)-\tilde{p}(t-1))}{p(t)}\geq \hat{v}_{t},$
% it's trivial. We now assume that $c\frac{\Delta}{\pi}\frac{(\tilde{p}(t)-\tilde{p}(t-1))}{p(t)}\leq \hat{v}_{t}.$
% As $g_{t}(v)$ is a concave increasing function in
% $[0,\hat{v}_{t}]$ and thus it's equivalent to show that $g_{t}(\bar{v}_{t})\leq g_{t}(c\frac{\Delta}{\pi}\frac{(\tilde{p}(t)-\tilde{p}(t-1))}{p(t)})$.
% For ease of presentation, we denote $k=g_{t}(\bar{v}_{t})=\frac{\Delta(\tilde{p}(t)-\tilde{p}(t-1))}{\pi}.$
% Then we have the following sequence of equivalent (or consequent) statements,
% \begin{eqnarray*}
%  & g_{t}(\bar{v}_{t})\leq g_{t}(c\frac{\Delta}{\pi}\frac{(\tilde{p}(t)-\tilde{p}(t-1))}{p(t)})\\
% \text{\ensuremath{\iff}} & k\leq g_{t}(\frac{ck}{p(t)})\\
% \stackrel{(a)}{\impliedby} & k\leq\frac{ck}{p(t)\hat{v}_{t}}g_{t}(\hat{v}_{t})\\
% \iff & c\geq\frac{g_{t}'(0)}{g_{t}(\hat{v}_{t})/\hat{v}_{t}},
% \end{eqnarray*}
% where the last inequality holds by the definition of $c$. For $(a),$ by the concavity of $g_{t}$, $g_{t}(0)=0$, and $\frac{ck}{p(t)}\leq \hat{v}_{t}$,
% we have $$g_{t}(\frac{ck}{p(t)})=g_t(\frac{ck}{p(t)\hat{v}_t}\hat{v}_t+(1-\frac{ck}{p(t)\hat{v}_t})0)\geq\frac{ck}{p(t)\hat{v}_{t}}g_{t}(\hat{v}_{t}).$$
% \fi

\subsection{Proof of Lemma~\ref{lem:bound_on_x_OOIC}}
\begin{proof}
For ease of presentation, define  
\[
x_p \triangleq\frac{\pi}{\Delta }\sum_{\left\{t:\;p\left(t\right)\leq p\right\}}g_t(\bar{v}_t).
\]
It is then equilivalent to show that $x_p\leq p$. Define $T_{1}\triangleq\min\{t:p(\tau)>p,\forall\tau\geq t\}-1$, i.e., for any $t>T_{1}$, we have $p(t)>p$, or equivalently if $p(t)\leq p,$ then $t\le T_{1}$. 
By definition, $x_{p}$ is determined by $\sigma^{[1:T_{1}]}$ only. Thus, in this proof, we only focus on the input horizon $t\in [T_{1}]$. 

We first consider a special case when $ p(t)\leq p,\;\forall t\in[T_{1}]$. By that $g_t(v),\forall t\in [T_1]$ are concave functions,  we have
$$\eta_{OPT}\left(\sigma^{[1:T_1]}\right)=\sum_{t=1}^{T_1}g_t\left(v_t^*\right)\leq \sum_{t=1}^{T_1} \left(g_t(0)+g'_t(0)v_t^*\right)=\sum_{t=1}^{T_1} p(t)v_t^*,$$
where $v_t^*,t\in[T_1]$ are the solution of the optimal offline algorithm under input $\sigma^{[1:T_1]}$. Then according to \eqref{eq:CR-Pursuit:aggregate.revenue} and that $\eta_{OPT}\left(\sigma^{[1:T_1]}\right)\leq p\cdot\Delta$, we have 
\begin{equation*}
x_{p} =\frac{\pi}{\Delta}\sum_{t=1}^{T_1}g_t(\bar{v}_t)
 =\frac{1}{\Delta}\eta_{OPT}\left(\sigma^{[1:T_1]}\right) \leq p.
\end{equation*}

We now consider the general cases, where there could be some slot(s) $\tau\in[T_1]$ such that $p(\tau)>p$. The we construct a new input sequence $\bar{\sigma}$ by interchange $g_\tau$ and $g_{\tau+1}$ in $\sigma$. Denote the input under $\bar{\sigma}$ as $\bar{g}_t$.  Let \textbf{$\bar{x}_{p}$, $\bar{p}(t)$} be the corresponding variables under $\bar{\sigma}$. 

To show that $x_p\leq p$, we first show $x_p\leq \bar{x}_p$. By definition, we observe that,
\[
\eta_{OPT}\left({\sigma}^{t}\right)=\eta_{OPT}\left(\bar{\sigma}^{t}\right), \forall t\leq \tau-1 \lor t \geq \tau+1;
\]
\[
g_t=\bar{g}_t, \forall t\leq \tau-1 \lor t \geq \tau+2.
\]
Besides, $g_\tau=\bar{g}_{\tau+1},g_{\tau+1}=\bar{g}_\tau.$
We discuss two cases.
\begin{itemize}
    \item  When $p(\tau+1)>p$: it is easy to see that $\bar{x}_p=x_p$.

 We then prove $x_p\leq p$ as follows: We continuously interchange with $p(\tau)>p$ with the input at its next slot until all the slots with $p(t)\leq p$ is at the front of it. At the meantime, $x_p$ keeps on non-decreasing.  Finally, we get a $\sigma'$, in which the price at each slot in $[T'_1]$ ($T'_1$ is corresponding to $T_1$ but defined under $\sigma'$) is less or equal to $p$, and $x_{p}\leq x_{p}'.$ Since in $\sigma'$, $p\geq p(t),\;\forall t$, from our analysis in the first part (special case), we have \textbf{$x_{p}'\leq p$}. It then follows that $x_{p}\leq p$.

    \item When $p(\tau+1)\leq p$: we have
\begin{align*}
{x}_{p}-\bar{x}_{p}= & \frac{\eta_{OPT}\left(\sigma^{[1:\tau+1]}\right)-\eta_{OPT}\left({\sigma}^{[1:\tau]}\right)}{\Delta}\\
 & -\frac{\eta_{OPT}\left(\bar{\sigma}^{[1:\tau]}\right)-\eta_{OPT}\left(\bar{\sigma}^{[1:\tau-1]}\right)}{\Delta}\\
 \stackrel{(a)}{\leq} & 0,
\end{align*}
where the step (a) is because of Lemma~\ref{lem:interchange_input_opt}.
\end{itemize}

Next, we prove $x_p\leq p$. We continuously interchange with $p(\tau)>p$ with the input at its next slot until all the slots with $p(t)\leq p$ is at the front of it. At the meantime, $x_p$ keeps on non-decreasing.  Finally, we get a $\sigma'$, in which the price at each slot in $[T'_1]$ ($T'_1$ is corresponding to $T_1$ but defined under $\sigma'$) is less or equal to $p$, and $x_{p}\leq x_{p}'.$ Since in $\sigma'$, $p\geq p(t),\;\forall t$, from our analysis in the first part (special case), we have \textbf{$x_{p}'\leq p$}. It then follows that $x_{p}\leq p$.

\end{proof}

\subsection{Proof of Lemma~\ref{lem:to_goal_OOIC}}
\begin{proof}

Suppose in $\sigma^{[1:T]}$, $p(t)$ takes $n$ different values, which are
denoted as $m\le p_{1}\leq p_{2}\leq\cdots\cdots\le p_{n}\le M$.
And define $y_{i}\triangleq\sum_{t,\;p(t)=p_{i}}\frac{\pi}{\Delta}g_t(\bar{v}_t)$.
Note that we have 
\[
\sum_{t=1}^{T}\frac{g_t(\bar{v}_t)}{p(t)}=\frac{\Delta}{\pi}\sum_{i=1}^{n}\frac{y_{i}}{p_{i}}.
\]
From Lemma \ref{lem:bound_on_x_OOIC}, we have $\sum_{j=1}^{i}y_{j}=x_{p_i}\leq p_{i}$.

Consider the following optimization problem:

\begin{eqnarray*}
\max &  & \sum_{i=1}^{n}\frac{y_{i}}{p_{i}}\\
s.t. &  & \sum_{j=1}^{i}y_{j}\leq p_{i},i\in[n]\\
 &  & y_{i}\geq0,i\in[n].
\end{eqnarray*}
The KKT conditions are sufficient and necessary conditions for optimality
for the above convex problem. Denote $\mu_{i}\ge0,i\in[n]$ as the
dual variables, then the KKT conditions can be expressed as:
\begin{align}
\frac{1}{p_{i}}-\sum_{j=1}^{n+1-i}\mu_{i} & =0,\forall i\in[n],\label{eq:KKT_1_OOIC}\\
\mu_{i}(p_{i}-\sum_{j=1}^{i}y_{j}) & =0,\forall i\in[n],\label{eq:KKT_2_OOIC}\\
\mu_{i} & \ge0,\forall i\in[n],\nonumber \\
y_{i} & \ge0,\forall i\in[n].\nonumber 
\end{align}
From (\ref{eq:KKT_1_OOIC}), we know that $\mu_{i}>0$ for all $i\in[n]$.
Thus from (\ref{eq:KKT_2_OOIC}), we have
\[
p_{i}-\sum_{j=1}^{i}y_{j}=0,\forall i\in[n].
\]
Thus we know the optimal primal solution is 
\[
y_{i}=p_{i}-p_{i-1},\forall i\in[n],
\]
where $p_{0}=0$. And the optimal objective value equals to $\sum_{i=1}^{n}\frac{p_{i}-p_{i-1}}{p_{i}}$. 

So 
\begin{eqnarray*}
\sum_{t=1}^{T}\frac{g_t(\bar{v}_t)}{p(t)} & = & \frac{\Delta}{\pi}\sum_{i=1}^{n}\frac{y_{i}}{p_{i}}\\
 & \leq & \frac{\Delta}{\pi}\sum_{i=1}^{n}\frac{p_{i}-p_{i-1}}{p_{i}}\\
 & = & \frac{\Delta}{\pi}\left(\frac{p_{1}}{p_{1}}+\sum_{i=2}^{n}\frac{p_{i}-p_{i-1}}{p_{i}}\right)\\
 & \leq & \frac{\Delta}{\pi}\left(1+\int_{p_{1}}^{p_{n}}\frac{1}{x}dx\right)\\
 & \leq &\frac{\Delta}{\pi}\left( 1+\ln\theta\right).
\end{eqnarray*}
This completes our proof.

\end{proof}

\subsection{Proof of Lemma~\ref{lem:increasing_price}}
\begin{proof}
We show that any input $\sigma^{[1:T]}$ is equivalent to (in the sense that the behaviors of both offline algorithm and the proposed online algorithm remain unchanged) an increasing price sequence as the following:
\begin{equation}
m\le p_{1}<p_{2}<\cdots<p_{n}\le M,\label{eq:increasing_price_seq}
\end{equation}
where $n\le T$. According to (\ref{eq:keep_cr_OWT}), \textsf{CR-Pursuit}$(\pi)$ will sell only when the current price is larger than the highest price in history. Thus for any input $\sigma^{[1:T]}$, we can delete the slots when \textsf{CR-Pursuit}$(\pi)$ does not sell, and the outputs of \textsf{CR-Pursuit}$(\pi)$ is then equivalent to the resulting increasing price sequence.
\end{proof}

\subsection{Proof of Lemma~\ref{upper_bound_v_t_OWTPE}}
\label{app:proof_upper_bound_v_t_OWTPE}
\begin{proof}

By Lemma \ref{lem:opt_difference_bound} and definition of $g_t(\bar{v}_t)$, we know that $\pi g_t(\bar{v}_t)=\eta_{OPT}\left(\sigma^{[1:t}\right)-\eta_{OPT}\left(\sigma^{[1:t-1]}\right)\leq g_t(\hat{v}_t)$ and then $\bar{v}_{t}\leq \hat{v}_{t}$, where $\hat{v}_t$ as the optimizer of $g_t(v_t)$. To simplify the explanation, let $k=g_{t}(\bar{v}_{t})$ and $\alpha=\frac{f\left(\hat{v}_t\right)}{\hat{v}_t}$. 

Define $\tilde{g}_t(v_t)= \left(p(t)-\alpha{v}_t\right){v}_t$.
By the convexity of $f_t(\cdot), f_t(0)=0$, we have $f_t(v_t)\leq \alpha v_t$,$\forall v_t\leq\hat{v}_t$. Then $g_t(v_t)\geq \tilde{g}_t(v_t),\forall v_t\leq\hat{v}_t$

Suppose $\tilde{v}_t$ is the smaller solution satisfying $\tilde{g}_t(\tilde{v}_t)=k$, i.e.,
\[
\tilde{v}_t=\frac{p(t)-\sqrt{p^2(t)-4\alpha k}}{2\alpha}=\frac{2k}{p(t)\left(1+\sqrt{1-\frac{4\alpha k}{p^2(t)}}\right)}.
\]

By observing
\[
k\pi\leq g_t(\hat{v}_t)=\tilde{g}_t(\hat{v}_t)\leq\frac{p^2(t)}{4\alpha},
\]
we have $\frac{4\alpha k }{p^2_t}\leq\frac{1}{\pi}$ (note that this also implies the existence of $\tilde{v}_t$). We then easily conclude 
\[
\tilde{v}_t\leq \frac{2k}{p(t)\left(1+\sqrt{1-\frac{1}{\pi}}\right)}.
\]

We claim $\bar{v}_t\leq\tilde{v}_t$. If $\tilde{v}_t> \hat{v}_t$, we have $\bar{v}_{t}\leq \hat{v}_{t}<\tilde{v}_t$; Otherwise, $\tilde{v}_t\leq\hat{v}_t$, we have
\[
k=g_t(\bar{v}_t)=\tilde{g}_t(\tilde{v}_t)\leq g_t{\left(\tilde{v}_t\right)}.
\]
Following $g_t({v_t})$ is increasing in $[0,\hat{v}_t]$, we conclude $\bar{v}_t\leq \tilde{v}_t$.

Finally, we conclude
\[
\bar{v}_t\leq \frac{2k}{p(t)\left(1+\sqrt{1-\frac{1}{\pi}}\right)} = \frac{2}{\left(1+\sqrt{1-\frac{1}{\pi}}\right)}\frac{g_t(\bar{v}_t)}{p(t)}.
\]

\end{proof}

\subsection{Proof of Lemma~\ref{lem:Phi_pi}}
\begin{proof}

From Lemma~\ref{upper_bound_v_t_OWTPE}, we have
\begin{align*}
\Phi_\Delta(\pi) & =\max_{\sigma^{[1:T]}}\sum_{t=1}^{T}v_t\le\frac{2}{\left(1+\sqrt{1-\frac{1}{\pi}}\right)}\sum_{t=1}^{T}\frac{g_t(\bar{v}_t)}{p(t)}.
\end{align*}
By Lemma \ref{lem:to_goal_OOIC}, 
%(Lemma \ref{lem:to_goal} and its proof can be found in the appendix)
 we know that 
\begin{equation}
\sum_{t=1}^{T}\frac{g_t(\bar{v}_t)}{p(t)}\leq\frac{\Delta}{\pi}(1+\ln\theta).\label{eq:goal}
\end{equation}
Then we can bound $\Phi_\Delta(\pi)$ as 
\begin{align*}
\Phi_\Delta(\pi) & \le\frac{2}{\left(1+\sqrt{1-\frac{1}{\pi}}\right)}\sum_{t=1}^{T}\frac{g_t(\bar{v}_t)}{p(t)}\\
 & \le\frac{2\Delta}{\pi\left(1+\sqrt{1-\frac{1}{\pi}}\right)}(1+\ln\theta).
\end{align*}
This completes our proof.

\iffalse
From Lemma~\ref{upper_bound_v_t_OWTPE}, we have
\begin{align*}
\Phi_\Delta(\pi) & =\max_{\sigma^{[1:T]}}\sum_{t=1}^{T}v_t\le\frac{2\Delta}{\pi\left(1+\sqrt{1-\frac{1}{\pi}}\right)}\sum_{t=1}^{T}\frac{\tilde{p}(t)-\tilde{p}(t-1)}{p(t)}.
\end{align*}
By Lemma \ref{lem:to_goal_OOIC}, 
%(Lemma \ref{lem:to_goal} and its proof can be found in the appendix)
 we know that 
\begin{equation}
\sum_{t=1}^{T}\frac{\tilde{p}(t)-\tilde{p}(t-1)}{p(t)}\leq1+\ln\theta.\label{eq:goal}
\end{equation}
Then we can bound $\Phi_\Delta(\pi)$ as 
\begin{align*}
\Phi_\Delta(\pi) & \le\frac{2\Delta}{\pi(1+\sqrt{1-\frac{1}{\pi}})}(\sum_{t=1}^{T}\frac{\tilde{p}(t)-\tilde{p}(t-1)}{p(t)})\\
 & \le\frac{2\Delta}{\pi(1+\sqrt{1-\frac{1}{\pi}})}(1+\ln\theta)=\bar{\Phi}(\pi).
\end{align*}
This completes our proof.

\fi
\end{proof}

\end{document}